\documentclass[11pt,a4paper]{article}
\usepackage[T1]{fontenc}
\usepackage[utf8]{inputenc}
\usepackage{lmodern}
\usepackage[margin=25mm]{geometry}
\usepackage{amsmath,amssymb,amsthm,mathtools,bm}
\usepackage{microtype,booktabs,array,graphicx,aliascnt,placeins,needspace,tocloft}
\usepackage[font=small,labelfont=bf,labelsep=period]{caption}
\usepackage[skins,breakable]{tcolorbox}
\newtcolorbox{resultbox}{enhanced,breakable,colback=black!2,colframe=black!35,
 boxrule=.5pt,arc=1pt,left=7pt,right=7pt,top=5pt,bottom=5pt,
 before skip=10pt,after skip=10pt}
\definecolor{examplebackground}{RGB}{242,249,243}
\definecolor{exampleframe}{RGB}{143,174,148}
\newtcolorbox{examplebox}{enhanced,colback=examplebackground,colframe=exampleframe,
 boxrule=.5pt,arc=1pt,left=7pt,right=7pt,top=5pt,bottom=5pt,
 before skip=10pt,after skip=10pt}
\usepackage[numbers,sort&compress]{natbib}

\usepackage[hidelinks]{hyperref}
\hypersetup{
 pdftitle={Quantum estimation, channel orders, and private capacity},
 pdfauthor={Christoph Hirche},
 pdfsubject={Quantum estimation, channel orders, and zero private capacity},
 pdfkeywords={quantum channels, private capacity, measured chi-square, MMSE, signed lifts, superactivation}
}
\usepackage[nameinlink,noabbrev]{cleveref}
\crefname{section}{Section}{Sections}
\crefname{appendix}{Appendix}{Appendices}
\crefname{theorem}{Theorem}{Theorems}
\crefname{proposition}{Proposition}{Propositions}
\crefname{lemma}{Lemma}{Lemmas}
\crefname{corollary}{Corollary}{Corollaries}
\crefname{definition}{Definition}{Definitions}
\crefname{example}{Example}{Examples}
\crefname{figure}{Figure}{Figures}
\newcommand{\equationalias}[1]{\csname ltx@label\endcsname{{#1}}}
\allowdisplaybreaks[2]
\numberwithin{equation}{section}
\newtheorem{theorem}{Theorem}[section]
\newaliascnt{proposition}{theorem}
\newtheorem{proposition}[proposition]{Proposition}\aliascntresetthe{proposition}
\newaliascnt{lemma}{theorem}
\newtheorem{lemma}[lemma]{Lemma}\aliascntresetthe{lemma}
\newaliascnt{corollary}{theorem}
\newtheorem{corollary}[corollary]{Corollary}\aliascntresetthe{corollary}
\theoremstyle{definition}
\newaliascnt{definition}{theorem}
\newtheorem{definition}[definition]{Definition}\aliascntresetthe{definition}
\newaliascnt{example}{theorem}
\newtheorem{example}[example]{Example}\aliascntresetthe{example}
\theoremstyle{remark}
\newaliascnt{remark}{theorem}
\aliascntresetthe{remark}
\DeclareMathOperator{\Tr}{Tr}
\DeclareMathOperator{\Herm}{Herm}
\DeclareMathOperator{\Var}{Var}
\DeclareMathOperator{\MSE}{MSE}
\DeclareMathOperator{\MMSE}{MMSE}
\DeclareMathOperator{\supp}{supp}
\DeclareMathOperator{\diag}{diag}
\DeclareMathOperator{\id}{id}
\newcommand{\cM}{\mathcal M}\newcommand{\cN}{\mathcal N}
\newcommand{\cD}{\mathcal D}\newcommand{\cJ}{\mathcal J}
\newcommand{\cS}{\mathcal S}\newcommand{\cV}{\mathcal V}
\newcommand{\cL}{\mathcal L}
\newcommand{\mchi}{\chi_{\mathrm m}^{2}}
\newcommand{\qchi}{\chi^{2}}
\newcommand{\schi}{\chi_{T_R}^{2}}
\newcommand{\ket}[1]{\lvert #1\rangle}\newcommand{\bra}[1]{\langle #1\rvert}
\newcommand{\norm}[1]{\lVert #1\rVert}
\newcommand{\Morder}{\cM\succeq_{\chi_{\mathrm m}^2}^{\mathrm c}\cN}
\newcommand{\Sorder}{\cM\succeq_{\chi_{T_R}^2}^{\mathrm c}\cN}
\newcommand{\Corder}{\cM\succeq_{\chi^2}^{\mathrm c}\cN}
\newcommand{\Dorder}{\cM\succeq_D^{\mathrm c}\cN}
\newcommand{\LNorder}{\cM\succeq_{\mathrm{ln}}^{\mathrm c}\cN}
\title{\LARGE Quantum estimation, channel orders, and private capacity}
\author{Christoph Hirche\\[5pt]
\normalsize Institute for Information Processing (tnt/L3S)\\[-1pt]
\normalsize Leibniz Universit\"at Hannover, Germany}
\date{}
\begin{document}
\maketitle
\vspace{-1.5em}
\begin{abstract}
Recent examples have shown that zero private capacity need not imply antidegradability and that two channels with zero private capacity can nevertheless transmit private information together. Existing entropic channel orders compare what a receiver and its environment can learn and thereby bound capacities. We connect these orders to the recent constructions through binary estimation, whose minimum mean-square error is governed by measured $\chi^2$ divergence. A new integral representation shows how measured $\chi^2$ on a qubit extension recovers quantum $\chi^2$. It lets us pass from complete measured-$\chi^2$ ordering to complete quantum-$\chi^2$, relative-entropy, and less-noisy ordering. Zhu and Wang used a signed lift to show that their qutrit channel has zero private capacity. We show that Hermitian-smoothed measured-$\chi^2$ ordering characterizes when such lifts exist, even with a quantum reference. Environmental dominance in binary estimation at every blocklength yields a finite-code reliability--secrecy bound. These results explain why complete comparison prevents activation with antidegradable helpers whereas regularized comparison alone does not. Building on that qutrit example, we establish this stability for a range of noisy Werner--Holevo channels, determine sharp private-capacity and antidegradability thresholds, and compute exact complementary capacities in a nondegradable range. By contrast, we extend Pauli half-erasure activation to all $1/2\le p<1$ and establish the same range for a new four-level family. Reference-assisted measured-$\chi^2$ witnesses show why these activating constructions lack complete comparison.
\end{abstract}

\setcounter{tocdepth}{2}
\begingroup\small\tableofcontents\endgroup
\clearpage

\section{Introduction}\label{sec:intro}
The unassisted private capacity $P(\cN)$ of a quantum channel characterizes the optimal rate at which classical information can be transmitted reliably while remaining inaccessible to the environment. Its single-use private information and regularized expressions are \cite{CaiWinterYeung2004,Devetak2005}
\begin{equation}
 \begin{aligned}
 P^{(1)}(\cN)&:=\max_{\{p_u,\rho_u^A\}}
 \bigl[I(U:B)-I(U:E)\bigr],\\
 P(\cN)&=\sup_{n\ge1}\frac1n P^{(1)}(\cN^{\otimes n}).
 \end{aligned}
 \label{eq:intro-private-capacity}
\end{equation}
Here $U$ is a finite classical label, and $B$ and $E$ are the receiver and complementary outputs. In $P^{(1)}(\cN^{\otimes n})$, the input states may be entangled across the $n$ uses. Evaluating these regularized expressions is generally difficult \cite{HircheLeditzky2023}. Even determining whether the private capacity is strictly positive can be challenging: a positive finite-blocklength private information certifies it, whereas a zero-capacity conclusion must exclude a private-information advantage at every blocklength. Part of the difficulty stems from strict superadditivity of private information over arbitrarily many channel uses \cite[Theorem~1]{ElkoussStrelchuk2015}. It is therefore useful to identify structural conditions that imply zero private capacity. A basic example is antidegradability: if the environment can simulate the receiver's output, every decoding procedure available to the receiver is also available to the environment, making secrecy at a positive rate impossible. The relation between a channel and its complement is therefore central to understanding its communication capabilities \cite{HircheLeditzky2023}.

However, state simulation is stronger than what is needed to recover a classical message. The private-capacity formula implies that a channel has zero private capacity exactly when its complement is less noisy at every blocklength \cite{Devetak2005}. Explicitly, the condition is
\begin{equation}
 I(U:E^n)\ge I(U:B^n)
 \quad\text{for every }n\ge1\text{ and every input ensemble on }A^n.
 \label{eq:intro-regularized-less-noisy}
\end{equation}
In the terminology of channel orders, this is the relation $\cN^c\succeq_{\mathrm{ln}}^{\mathrm{reg}}\cN$, corresponding to the complementary version of Watanabe's less-noisy class \cite{watanabe2012private,HircheLeditzky2023}. Entangled block inputs are allowed, and determining membership requires control of arbitrarily many channel uses. This motivates intermediate conditions that are weaker than antidegradability but simpler to verify than the regularized comparison.

One such condition is that the complement dominates the channel in the completely less-noisy order, $\cN^c\succeq_{\mathrm{ln}}^{\mathrm c}\cN$ \cite{HircheRouzeFranca2022}. We then call $\cN$ completely anti-less noisy. For every input ensemble on $RA$, with an arbitrary quantum reference $R$ that may be correlated with the input, this means
\begin{equation}
 I(U:RE)\ge I(U:RB).
 \label{eq:intro-complete-less-noisy}
\end{equation}
The complete less-noisy order tensorizes by \cite[Lemma~3.1]{HircheRouzeFranca2022}, and therefore implies the corresponding comparison at every blocklength. This removes the need to check arbitrarily many uses, at the cost of a single-use comparison with reference systems of arbitrarily large finite dimension.

Whether zero private capacity can occur beyond antidegradability was a longstanding open question \cite{BuscemiDattaStrelchuk2014}. In very recent work, Zhu and Wang \cite{ZhuWang2026} resolved it by giving a qutrit channel $\Lambda_{1/2}$ with zero private and quantum capacities that is neither antidegradable nor PPT. Its complement nevertheless satisfies the complete less-noisy comparison. Their proof relaxes state simulation to observable estimation: one linear map reproduces the receiver's observable means from the environment with no larger variance, and does so in the presence of an arbitrary reference. Measuring the corresponding environment observable therefore gives an unbiased sample-mean estimator with no larger variance \cite[Remark~3.8]{ZhuWang2026}. Zhu and Wang call such a map a complete variance-dominating signed lift \cite[Definition~3.3]{ZhuWang2026}. It need not map states to states, but it preserves the information needed for these estimation tasks. This new criterion provides the starting point for the comparisons developed here.

Another important question is superactivation: can two channels that individually have zero capacity have positive capacity jointly? For the quantum capacity, Smith and Yard combined channels from the antidegradable and PPT families to obtain such a pair \cite{SmithYard2008}. Both families are closed under tensor products. Here PPT means that output transposition composed with the channel is completely positive. Private capacity was already known to be nonadditive \cite{LiWinterZouGuo2009,SmithSmolin2009}, but its superactivation remained open \cite{HircheLeditzky2023}. In a separate recent paper, Zhu and Wang \cite{ZhuWangActivation2026} proved that a pair of channels can satisfy
\begin{equation}
 P(\cN)=P(\cM)=0,\qquad P(\cN\otimes\cM)>0.
 \label{eq:intro-superactivation}
\end{equation}
Their construction uses a transpose-antidegradable channel: for a CPTP map $\cD$, one has $\cD\circ\cN^c=T_B\circ\cN$, where $T_B$ is output transposition \cite{BradlerDutilHaydenMuhammad2010,ZhuWangActivation2026}. Shortly afterwards, a construction using two-qubit Pauli channels appeared, including a variant with a half-erasure helper \cite{privateCapZero2026}.

The initial qutrit channel $\Lambda_{1/2}$ of Zhu and Wang \cite[Theorems~1.1 and~4.3]{ZhuWang2026}, defined here in \eqref{eq:wh}, cannot activate with any antidegradable helper $\cM$, because antidegradability implies completely anti-less-noisy ordering. The two complete comparisons tensorize, giving
\begin{equation}
 \Lambda_{1/2}^c\succeq_{\mathrm{ln}}^{\mathrm c}\Lambda_{1/2}
 \quad\Longrightarrow\quad P(\Lambda_{1/2}\otimes\cM)=0
 \quad(\cM\text{ antidegradable}).
 \label{eq:intro-qutrit-stability}
\end{equation}
In contrast, the four-level channel $\cN_4$ from their activation paper \cite[Theorem~1]{ZhuWangActivation2026} satisfies
\begin{equation}
 \begin{gathered}
 P(\cN_4\otimes\mathcal E_{2,1/2})>0,\\
 \cN_4^c\succeq_{\mathrm{ln}}^{\mathrm{reg}}\cN_4,
 \qquad \cN_4^c\not\succeq_{\mathrm{ln}}^{\mathrm c}\cN_4.
 \end{gathered}
 \label{eq:intro-activation-separation}
\end{equation}
Here $\mathcal E_{2,1/2}$ is the qubit erasure channel with erasure probability $1/2$. The regularized relation follows from $P(\cN_4)=0$, whereas a complete relation would force zero product capacity with this antidegradable helper. These examples therefore separate complete from regularized anti-less-noisiness. Together, these developments describe the present state of the art and mark substantial progress over the past few months: zero private capacity now admits mechanisms beyond physical simulation, and channels with zero private capacity can nevertheless be useful jointly.

\vspace{\baselineskip}

Our aim is to place these recent examples and the signed-lift criterion in a common framework with the previous entropic partial orders \cite{watanabe2012private,HircheRouzeFranca2022}. \Cref{fig:classes} summarizes the resulting channel comparisons.
We begin with an operational estimation task. Consider a uniformly distributed bit $U\in\{+1,-1\}$ encoded in two states on a receiver system $B$, as illustrated in \cref{fig:mmse}. The minimum mean-square error $\MMSE(U|B)$ measures how accurately the bit can be inferred from $B$ using a real-valued estimate. This is a binary instance of Personick's quantum quadratic-estimation problem \cite{Personick1971}. We prove in \cref{thm:privacy} that the following comparison suffices for zero private capacity:
\begin{equation}
 \MMSE(U|E^n)\le\MMSE(U|B^n)
 \quad\text{for every binary encoding and every }n\ge1.
 \label{eq:intro-estimation-comparison}
\end{equation}
To see why this suffices, note that a binary decoder with error $\eta$ is an estimate in $\{+1,-1\}$ with mean-square error $4\eta$. The comparison gives the environment an estimator whose error also vanishes when decoding becomes reliable; taking its sign yields a reliable binary decision. Secrecy instead requires the environment's optimal bit-guessing probability to approach $1/2$. Applying this pairwise observation to all messages gives a quantitative reliability--secrecy bound for any finite code, as shown in \cref{thm:privacy}.

Next, we formulate the estimation comparison using the measured $\chi^2$ divergence. For output states $\omega_+,\omega_-$ and their average $\bar\omega=(\omega_++\omega_-)/2$, the binary MMSE identity of \cite[Lemma~25]{MandalPfister2026} is stated in SLD $\chi^2$ notation. Since SLD $\chi^2$ equals measured $\chi^2$ \cite[Lemma~5]{TemmeVerstraete2014} (see also \eqref{eq:quadratic}), it reads
\begin{equation}
 1-\MMSE(U|B)=\mchi(\omega_+\Vert\bar\omega).
 \label{eq:intro-mmse}
\end{equation}
Our first result shows that complete measured-$\chi^2$ ordering implies the quantum $\chi^2$ ordering of \cite[Section~2.6]{HircheTomamichel2024}, with arbitrary quantum references. Its quadratic form is the Bogoliubov--Kubo--Mori (BKM) metric, so we also refer to this comparison as BKM ordering (see \cref{sec:bkm}). It is equivalent to complete relative-entropy and less-noisy ordering \cite{HircheRouzeFranca2022}. The hierarchy is
\begin{equation}
 \cM\succeq_{\chi_{\mathrm m}^2}^{\mathrm c}\cN
 \ \Longrightarrow\ \cM\succeq_{\chi^2}^{\mathrm c}\cN
 \ \Longleftrightarrow\ \cM\succeq_D^{\mathrm c}\cN
 \ \Longleftrightarrow\ \cM\succeq_{\mathrm{ln}}^{\mathrm c}\cN.
 \label{eq:intro-order-hierarchy}
\end{equation}
The first implication follows from the coherent-qubit realization in \cref{prop:qubit-realization}. The subsequent equivalences connect it to the established entropic framework \cite{BelzigGaoSmithWu2025,HircheRouzeFranca2022}. Our second result characterizes when $\cN$ has a complete variance-dominating signed lift through $\cM$ by minimizing measured $\chi^2$ over trace-one Hermitian operators preserving a fixed classical--quantum measurement record. This identifies the statistical information that one common estimator map must preserve. For a CPTP map $\cD$, we write $\cM\succeq_{\chi_{T_R}^2}^{\mathrm c}\cN$ for the resulting Hermitian-smoothed comparison between the two channels. We obtain
\vspace{1mm}
\begin{equation}
 \begin{array}{c@{\;}c@{\;}c@{\;}c@{\;}c}
 \cN=\cD\circ\cM &\Longrightarrow&
 \cM\succeq_{\chi_{T_R}^2}^{\mathrm c}\cN &\Longrightarrow&
 \cM\succeq_{\mchi}^{\mathrm c}\cN.\\[-1pt]
 &&\Updownarrow&&\\[-1pt]
 &&\text{complete signed lift exists}&&
 \end{array}
 \label{eq:intro-smoothing-order}
\end{equation}
\vspace{1mm}

A further question is whether a comparison remains valid when the channel is paired with another channel. Different complete criteria can be combined by first passing to a common order implied by both premises. For example,
\begin{equation}
 \begin{gathered}
 \cM_1\succeq_{\mchi}^{\mathrm c}\cN_1,\qquad
 \cM_2\succeq_D^{\mathrm c}\cN_2\\
 \Longrightarrow\quad
 \cM_1\otimes\cM_2\succeq_D^{\mathrm c}\cN_1\otimes\cN_2.
 \end{gathered}
 \label{eq:intro-mixed-tensor}
\end{equation}
The first premise implies complete relative-entropy ordering by \eqref{eq:intro-order-hierarchy}. That common order then tensorizes \cite[Lemma~3.1]{HircheRouzeFranca2022}. A regularized comparison, required only at every tensor power of the channel itself, need not be stable under pairing with another channel.

The same orders have a capacity consequence when the receiver is more informative than the environment. We write $Q(\cN)$ for its quantum capacity, the optimal rate of reliable quantum communication \cite[Chapter~24]{Wilde2017}. Receiver dominance at every blocklength makes the complement's private capacity vanish. The complementarity bounds of \cite{HircheLeditzky2023} then imply (\cref{prop:single-letter})
\begin{equation}
 \cN\succeq_{\mathrm{ln}}^{\mathrm{reg}}\cN^c
 \quad\Longrightarrow\quad
 P(\cN^c)=0,\qquad
 P(\cN)=P^{(1)}(\cN)=Q(\cN)=Q^{(1)}(\cN).
 \label{eq:intro-receiver-capacities}
\end{equation}
Stronger complete orders, such as $\cN\succeq_{\chi_{\mathrm m}^2}^{\mathrm c}\cN^c$, imply this premise by \eqref{eq:intro-order-hierarchy} and tensorization.

The examples in \cref{sec:examples} build on the three recent constructions discussed above \cite{ZhuWang2026,ZhuWangActivation2026,privateCapZero2026} and extend their channel families and encodings. The qutrit-derived channels in \cref{sec:wh,sec:qubitqutrit}, including the mixtures of \cref{ex:mixtures}, retain complete comparison without antidegradability. The Pauli and coherent-ladder channels in \cref{sec:pauli,sec:coherent} have zero private capacity but activate with a qubit erasure channel. A compact fixed example is the two-qubit Pauli channel $\mathcal F$ of \cite{privateCapZero2026}. With adjacent Pauli symbols denoting tensor products, \cref{ex:pauli-erasure} extends its half-erasure activation to every $1/2\le p<1$:
\begin{equation}
 \begin{gathered}
 \mathcal F(X)=\frac16\sum_{V\in\mathcal P}VXV^\dagger,
 \qquad\mathcal P=\{XX,YI,IZ,XZ,IY,YZ\},\\[2pt]
 P(\mathcal F)=P(\mathcal E_{2,p})=0,
 \qquad P(\mathcal F\otimes\mathcal E_{2,p})
 \ge\frac{\ln2\,(1-p)^2}{216p(1+\sqrt3)}>0.
 \end{gathered}
 \label{eq:intro-activation}
\end{equation}
Our coherent-ladder family also activates throughout this erasure interval for every finite $s\ge3$ (\cref{thm:coherent-family}). At $s=4$ it recovers the channel of Zhu and Wang, whose activation throughout this interval was already proved in \cite[Theorem~1]{ZhuWangActivation2026}. Varying the transition amplitudes and encoding together preserves the error-detecting interference.

The Werner--Holevo complements illustrate the single-letter consequence above. For $1/2\le x\le1$, \cref{cor:wh-complement} gives
\begin{equation}
 P(\Lambda_x^c)=Q(\Lambda_x^c)=h_2(x)-(1-x)\log_2 3
 \quad\text{bits per use},
 \label{eq:intro-wh-complement-capacity}
\end{equation}
For $1/2\le x<4/7$, this equality holds although $\Lambda_x^c$ is nondegradable. The complete comparison for $\Lambda_x$ makes both complementary capacities equal to one-use coherent information by \cref{prop:single-letter} and the complementarity bounds of \cite{watanabe2012private,HircheLeditzky2023}. Concavity and covariance then make the maximally mixed qutrit input optimal.

The remainder of this paper is organized as follows. \Cref{sec:estimation} develops binary estimation and Hermitian smoothing. \Cref{sec:orders} establishes the channel-order hierarchy, relates complete and regularized comparisons, and gives single-letter capacity consequences when the receiver dominates the environment. \Cref{sec:privacy} proves a reliability--secrecy bound and characterizes zero private capacity. \Cref{sec:examples} collects the channel examples, followed by the discussion in \cref{sec:discussion}. The three appendices collect the estimation and comparison proofs, certificates for the non-antidegradable examples, and calculations for the activation examples.

\subsection{Notation and preliminaries}
All systems are associated with finite-dimensional Hilbert spaces, with $d_A=\dim A$ denoting the dimension of system $A$. We write $\cL(A)$ and $\Herm(A)$ for linear and Hermitian operators on $A$, respectively; $\operatorname{im}T$ denotes the image of a linear map $T$. A quantum state is a positive semidefinite operator of trace one, and it is faithful if it is positive definite. A quantum channel $\Phi:A\to B$ is a completely positive trace-preserving (CPTP) map, with Hilbert--Schmidt adjoint $\Phi^\dagger$. A Stinespring isometry $V:A\to B\otimes E$ defines $\Phi(X)=\Tr_E(VXV^\dagger)$ and a complementary channel $\Phi^c(X)=\Tr_B(VXV^\dagger)$. Complementary channels are equivalent up to output isometries on their supports. For a state $\tau$, the replacer channel is $\mathcal R_\tau(X)=\Tr(X)\tau$.

We compare channels $\cM_{A\to E}$ and $\cN_{A\to B}$. Input and output subscripts are omitted once the systems are clear. When a channel acts on a joint state, it acts only on its specified input; identities on all other systems are implicit. The same convention applies to adjoints and linear observable maps. Antidegradability of $\cN$ means $\cN=\cD\circ\cN^c$ for a CPTP map $\cD$. We write $P(\cN)$ and $Q(\cN)$ for its private classical and quantum capacities. Their precise operational definitions can also be found in \cite[Chapters~23 and~24]{Wilde2017}. They satisfy $Q(\cN)\le P(\cN)$ \cite{Devetak2005}. The erasure channel $\mathcal E_{d,p}$ transmits its $d$-dimensional input with probability $1-p$ and otherwise produces an orthogonal flag $\ket e$. On an input operator $X$, it acts as
\begin{equation}
 \mathcal E_{d,p}(X)=(1-p)X\oplus p\Tr(X)\ket e\bra e.
 \label{eq:erasure-definition}
\end{equation}
The two blocks are the transmitted input and the one-dimensional erasure sector. The channel is antidegradable for $p\ge1/2$.

Entropies, relative entropies, mutual informations, and capacities are expressed in bits. We write $\log_2$ for the base-two logarithm and $\ln$ for the natural logarithm. The binary entropy is $h_2(p):=-p\log_2p-(1-p)\log_2(1-p)$, with continuous values at $p=0,1$. Our entropy and information conventions are
\begin{equation}
 \begin{aligned}
 S(\rho)&=-\Tr\rho\log_2\rho,&
 D(\rho\Vert\sigma)&=\Tr\rho(\log_2\rho-\log_2\sigma),\\
 I(A:B)&=S(A)+S(B)-S(AB),&
 I(A:B|R)&=S(AR)+S(BR)-S(R)-S(ABR).
 \end{aligned}
 \label{eq:entropic-notation}
\end{equation}
We set $D(\rho\Vert\sigma)=+\infty$ when $\supp\rho\not\subseteq\supp\sigma$.
For a Hermitian observable $G$, the variance is
\begin{equation}
 \Var_\rho(G)=\Tr\rho G^2-(\Tr\rho G)^2.
 \label{eq:variance}
\end{equation}
Throughout, inverses of positive semidefinite operators and positive superoperators mean their Moore--Penrose pseudoinverses \cite{Penrose1955}: the inverse on the support, extended by zero on the kernel. A quadratic form defined by an inverse has value $+\infty$ when its direction lies outside the range of the corresponding superoperator. For a channel $\Phi$, its effective output space is $\supp\Phi(I)$; full rank on an output space means full rank relative to this support.

\begin{figure}[t]
\centering\includegraphics[width=\textwidth]{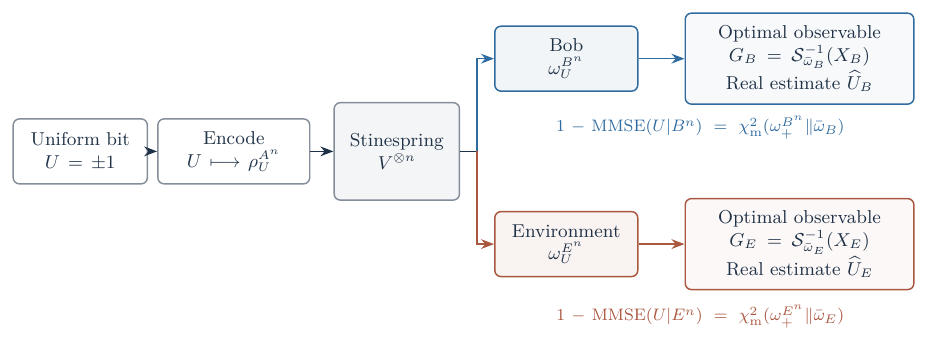}
\caption{Binary estimation at the receiver and the environment of a single Stinespring dilation. Each output admits a real-valued score whose optimal mean-square error is given by measured $\chi^2$. The averages $\bar\omega_B,\bar\omega_E$ and differences $X_B,X_E$ are formed from the two output states; the SLD map and its quadratic form are defined in \eqref{eq:quadratic}.}
\label{fig:mmse}
\end{figure}

\section{Quantum estimation}\label{sec:estimation}
We develop the estimation formulation first, then restrict the observable scores to a fixed measurement record. These two tasks provide the divergences used in the channel comparisons.

\subsection{Binary estimation}\label{sec:mmse}
We first formulate the binary estimation task underlying the channel comparisons, as illustrated in \cref{fig:mmse}. Consider an equiprobable ensemble $\{\omega_+,\omega_-\}$ on a receiver system $B$. The sender chooses $U\in\{+1,-1\}$, and the receiver performs a finite-outcome positive operator-valued measure (POVM) $\{M_y\}_{y\in\mathcal Y}$ with real-valued estimates $g_y$. For outcome $y\in\mathcal Y$, the estimator $\widehat U_{M,g}$ reports $g_y$. Its mean-square error is
\begin{equation}
 \MSE(U,\widehat U_{M,g}):=\mathbb E[(U-\widehat U_{M,g})^2]
 =\frac12\sum_{u=\pm1,\,y\in\mathcal Y}\Tr(\omega_uM_y)(g_y-u)^2.
 \label{eq:risk-povm}
\end{equation}
The minimum mean-square error is
\begin{equation}
 \MMSE(U|B):=\inf_{\substack{\{M_y\}\ \mathrm{POVM}\\g_y\in\mathbb R}}\MSE(U,\widehat U_{M,g}),
 \label{eq:mmse-definition}
\end{equation}
where the infimum ranges over all finite-outcome POVMs on $B$ and corresponding real estimates. Although $U$ is binary, \cref{prop:mmse} gives a spectral optimum with at most $\dim B$ outcomes.

For states $\rho,\sigma$, the measured $\chi^2$ divergence optimizes the classical $\chi^2$ divergence over POVMs \cite{TemmeVerstraete2014}:
\begin{equation}
 \mchi(\rho\Vert\sigma):=\sup_{\{M_y\}}\sum_y
 \frac{[\Tr M_y(\rho-\sigma)]^2}{\Tr M_y\sigma}.
 \label{eq:measured-definition}
\end{equation}
As usual, a term with zero denominator is zero if its numerator vanishes and $+\infty$ otherwise.
To connect this definition to the binary estimation task, put $\sigma=(\omega_++\omega_-)/2$ and $X=(\omega_+-\omega_-)/2$. Define the SLD map and its quadratic form by
\begin{equation}
 \begin{aligned}
 \cS_\sigma(G)&=\frac{\sigma G+G\sigma}{2},\\
 q_\sigma(X)&:=\sup_{G\in\Herm(B)}\{2\Tr XG-\Tr\sigma G^2\}\\
 &=\begin{cases}
 \Tr X\cS_\sigma^{-1}(X),&X\in\operatorname{im}\cS_\sigma,\\
 +\infty,&\text{otherwise}.
 \end{cases}
 \end{aligned}
 \label{eq:quadratic}
\end{equation}
The SLD/Bures representation of measured $\chi^2$ is \cite[Lemma~5]{TemmeVerstraete2014}
\[
 \mchi(\rho\Vert\sigma)=q_\sigma(\rho-\sigma).
\]
Its variational expression is the order-two measured R\'enyi formula after shifting the score by $I$ \cite[Lemma~3 and Theorem~4]{BertaFawziTomamichel2017}. For the binary pair above, $X\in\operatorname{im}\cS_\sigma$, and the optimal score $G_*=\cS_\sigma^{-1}(X)$ solves the SLD equation $\cS_\sigma(G_*)=X$. The next proposition restates the binary MMSE result of \cite[Lemmas~7 and~25]{MandalPfister2026}. For completeness, we give a short proof of the identity and its optimal estimator.

\Needspace{10\baselineskip}
\begin{resultbox}
\begin{proposition}[Binary MMSE and measured $\chi^2$]\label{prop:mmse}
The minimum mean-square error satisfies
\begin{equation}
 1-\MMSE(U|B)=\mchi(\omega_+\Vert\sigma).
 \label{eq:mmse}
\end{equation}
An optimal strategy measures $G_*=\cS_\sigma^{-1}(X)$ and reports its eigenvalue.
\end{proposition}
\end{resultbox}
\begin{proof}
For a POVM estimator set $G=\sum_yg_yM_y$. Positivity of the effects gives the operator inequality
\begin{equation}
 \sum_yg_y^2M_y-G^2
 =\sum_y(g_yI-G)M_y(g_yI-G)\succeq0.
 \label{eq:povm-variance}
\end{equation}
Hence the spectral measurement of $G$ cannot increase the mean-square error. Let $\widehat U_G$ denote its reported eigenvalue. For that estimator,
\begin{equation}
 \MSE(U,\widehat U_G)=\tfrac12\Tr\omega_+(G-I)^2+\tfrac12\Tr\omega_-(G+I)^2
 =1-2\Tr XG+\Tr\sigma G^2.
 \label{eq:risk-score}
\end{equation}
Optimizing \eqref{eq:risk-score} and using $\mchi(\omega_+\Vert\sigma)=q_\sigma(X)$ proves \eqref{eq:mmse}.
\end{proof}
The identity is a binary instance of Personick's quantum quadratic estimation \cite{Personick1971}. For any fixed POVM, let $p_u(y)=\Tr(\omega_uM_y)$ and $q_y=\Tr(\sigma M_y)=[p_+(y)+p_-(y)]/2$. The posterior-mean estimate at outcome $y$ is $[p_+(y)-p_-(y)]/(2q_y)$ when $q_y>0$. The constant estimate zero has mean-square error $1$, while these posterior means reduce the error by $\sum_{y:q_y>0}[p_+(y)-p_-(y)]^2/(4q_y)$. Taking the supremum over POVMs gives \eqref{eq:mmse}. Thus the denominator in \eqref{eq:measured-definition} is the outcome probability $q_y$ of the average state $\sigma$.

\begin{lemma}[Comparison with trace distance]\label{lem:distance}
For $d=\tfrac12\norm{\omega_+-\omega_-}_1$,
\begin{equation}
 d^2\le\mchi(\omega_+\Vert\sigma)\le d.
 \label{eq:distance}
\end{equation}
\end{lemma}
\begin{proof}
For an arbitrary POVM $\{M_y\}$, set $p_u(y)=\Tr(\omega_uM_y)$, $q_y=(p_+(y)+p_-(y))/2$, and $x_y=(p_+(y)-p_-(y))/2$. Terms with $q_y=0$ also have $x_y=0$ and are omitted below. Since $|x_y|\le q_y$, Cauchy--Schwarz gives
$(\sum_y|x_y|)^2\le\sum_yx_y^2/q_y\le\sum_y|x_y|$.
A Helstrom-optimal binary discrimination POVM has $\sum_y|x_y|=d$, so the first inequality proves the lower bound. The second holds for every POVM, and taking the supremum proves the upper bound.
\end{proof}

\FloatBarrier
\subsection{Hermitian smoothing}\label{sec:smoothing}
The measured-$\chi^2$ quadratic form extends to Hermitian first arguments by $\mchi(H\Vert\sigma)=q_\sigma(H-\sigma)$. This extension can be tested using physical states. For faithful $\sigma$ and traceless Hermitian $X$, the operator $\sigma+tX$ is a state for sufficiently small $|t|$, and
\begin{equation}
 \mchi(\Phi(\sigma+tX)\Vert\Phi(\sigma))
 =t^2q_{\Phi(\sigma)}(\Phi(X)).
 \label{eq:tangent}
\end{equation}
Thus a channel comparison on all state pairs also compares every traceless Hermitian direction. For arbitrary Hermitian $X$, write $X=X_0+c\sigma$, where $\Tr X_0=0$ and $c=\Tr X$. Since $\cS_\sigma^{-1}(\sigma)=I$, we have $q_\sigma(X)=q_\sigma(X_0)+c^2$, with the same additional term at both channel outputs. For a smooth state family through $\sigma$ with derivative $X$, this quadratic form is the SLD Fisher information \cite{BraunsteinCaves1994}.

We now restrict the estimator to observables whose expectation is determined by a fixed classical--quantum measurement record. The corresponding divergence admits a minimization over Hermitian completions of that record. Choose an orthonormal basis of the reference system $R$ to define partial transposition $T_R$, and set
\begin{equation}
 \begin{gathered}
 \cV_R=\{G\in\Herm(RA):G^{T_R}=G\},\qquad P_R=\tfrac12(\id+T_R),\\
 \schi(\rho\Vert\sigma)
 =\sup_{G\in\cV_R}\{2\Tr(\rho-\sigma)G-\Tr\sigma G^2\}.
 \end{gathered}
 \label{eq:restricted}
\end{equation}
Here $P_R$ is the Hilbert--Schmidt orthogonal projection onto $\cV_R$ as a real vector space. It need not be positive. The subscript $T_R$ indicates that admissible scores obey $G^{T_R}=G$ on the reference $R$. For a fixed state pair, the restricted value may depend on the chosen basis, since that choice determines the measurement record. The complete channel order introduced below is basis independent: it tests every reference input state, so a basis change can be absorbed by unitary conjugation on $R$. For a trivial reference, $\schi$ reduces to ordinary measured $\chi^2$.

For an equiprobable binary ensemble $\{\omega_+,\omega_-\}$ on $RA$, define the restricted minimum mean-square error by
\begin{equation}
 \begin{aligned}
 \MMSE_{\cV_R}(U|RA)
 &:=\inf_{G\in\cV_R}\MSE(U,\widehat U_G)\\
 &=\inf_{G\in\cV_R}\frac12\sum_{u=\pm1}\Tr\omega_u(G-uI)^2.
 \end{aligned}
 \label{eq:restricted-mmse-definition}
\end{equation}
Here $\widehat U_G$ measures $G$ spectrally and reports its eigenvalue, so the right-hand expression is \eqref{eq:risk-povm} restricted to such measurements. Its spectral projectors need not lie in $\cV_R$.

The next theorem identifies the restricted score with the least measured $\chi^2$ among Hermitian completions of the same measurement record.
\begin{resultbox}
\begin{theorem}[Fixed-record smoothing]\label{thm:completion}
For a faithful $\sigma$,
\begin{equation}
 \schi(\rho\Vert\sigma)
 =\min_{\substack{H\in\Herm(RA)\\P_RH=P_R\rho}}\mchi(H\Vert\sigma).
 \label{eq:completion}\equationalias{eq:intro-smoothing}
\end{equation}
Every feasible $H$ has trace one. When the denominator is the average of an equiprobable binary ensemble, the value equals $1-\MMSE_{\cV_R}(U|RA)$.
\end{theorem}
\end{resultbox}
\begin{proof}
Set $X=\rho-\sigma$. The first-order optimality condition for the restricted maximizer is $P_R\cS_\sigma(G_*)=P_RX$. Hence $K_*:=\cS_\sigma(G_*)-X$ belongs to $\cV_R^\perp$. For any $K\in\cV_R^\perp$ and $G\in\cV_R$, we have $\Tr KG=0$, so the variational formula for $q_\sigma(X+K)$ is bounded below by the restricted optimum. At $K=K_*$, the observable $G_*$ also maximizes the unrestricted objective, proving equality. Since every $K\in\cV_R^\perp$ is traceless, the feasible completions have trace one. The estimation identity follows from \eqref{eq:risk-score}.
\end{proof}

\paragraph{Measurement records.}
With $d_R=\dim R$, choose the real POVM
\begin{equation}
 E_i=\frac{\ket i\bra i}{d_R},\qquad
 E_{ij,\pm}=\frac{\ket{v_{ij,\pm}}\bra{v_{ij,\pm}}}{d_R},\qquad
 v_{ij,\pm}=\frac{\ket i\pm\ket j}{\sqrt2}\quad(i<j).
 \label{eq:real-povm}
\end{equation}
Its effects sum to $I$ and span the real symmetric matrices. The measurement channel stores each outcome in a separate classical register $Z$:
\begin{equation}
 \begin{aligned}
 \Theta_R(W)&=\sum_z\Tr(E_zW)\ket z\bra z_Z,\\
 (\Theta_R\otimes\id_A)(\rho_{RA})
 &=\sum_z\ket z\bra z_Z\otimes
 \Tr_R[(E_z\otimes I_A)\rho_{RA}].
 \end{aligned}
 \label{eq:record-channel}
\end{equation}
The following lemma identifies this classical--quantum record with the projection constraint.
\begin{lemma}[Fixed-record constraint]\label{lem:record}
For Hermitian operators $H,\rho$ on $RA$,
\begin{equation}
 P_RH=P_R\rho
 \quad\Longleftrightarrow\quad
 (\Theta_R\otimes\id_A)(H)=(\Theta_R\otimes\id_A)(\rho).
 \label{eq:record}
\end{equation}
\end{lemma}
\begin{proof}
Put $\Delta=H-\rho$. Equality of the records means $\Tr_R[(E_z\otimes I_A)\Delta]=0$ for every outcome $z$. Since the effects span the real symmetric operators on $R$, these conditions are equivalent to Hilbert--Schmidt orthogonality of $\Delta$ to $\cV_R$, or $P_R\Delta=0$.
\end{proof}
Thus the constraint fixes every outcome probability and unnormalized conditional state produced by the specified real measurement on $R$. For a qubit reference, it fixes the coefficient operators of $I_R,X_R,Z_R$ and leaves the $Y_R$ component unconstrained. Equation~\eqref{eq:completion} minimizes distinguishability over all trace-one Hermitian operators consistent with this one record. Locally, it minimizes SLD information over Hermitian perturbations invisible to the record channel, namely those in $\ker(\Theta_R\otimes\id_A)$. These play the role of nuisance directions in quantum semiparametric estimation \cite{TsangAlbarelliDatta2020}.

The completion need not be positive. For the Bell state $\rho=\ket{\Phi^+}\bra{\Phi^+}$ and $\sigma=I/4$, an optimal completion is $(I+X\otimes X+Z\otimes Z)/4$, with one eigenvalue $-1/4$; the restricted and unrestricted divergences are $2$ and $3$, respectively. The restricted task also differs from first applying $\Theta_R$: an admissible observable such as $X\otimes X+Z\otimes Z$ can have spectral projections outside $\cV_R$. The restriction concerns the score observable, rather than every POVM effect.

Hermitian optimizers also arise in the variational form of the measured smooth collision divergence of \cite[Theorem~2]{RegulaTomamichel2026}. There, a feasible Hermitian $H$ obeys $H\preceq\rho$ and $\Tr(\rho-H)_+\le\varepsilon$. Here, the completion has trace one and preserves a specified measurement record. A related approach fixes a marginal during partial smoothing \cite{AnshuBertaJainTomamichel2018}.

\section{Quantum channel orders}\label{sec:orders}
We first define complete channel comparisons, prove their hierarchy, and relate them to regularized comparisons. A coherent-qubit realization then connects measured and quantum $\chi^2$. Finally, receiver dominance gives single-letter capacity consequences.

\subsection{Definitions and hierarchy}\label{subsec:hierarchy}
The estimation and smoothing tasks define comparisons between channels $\cM_{A\to E}$ and $\cN_{A\to B}$. For a finite reference $R$, let $\rho,\sigma$ be states on $RA$. We write $\Morder$ if
\begin{equation}
 \mchi(\cN(\rho)\Vert\cN(\sigma))
 \le\mchi(\cM(\rho)\Vert\cM(\sigma))
 \quad\text{for every }R,\rho,\sigma.
 \label{eq:complete-measured}
\end{equation}
Define $\Sorder$, $\Corder$, and $\Dorder$ similarly for the restricted $\chi_{T_R}^2$, quantum $\chi^2$, and relative-entropy quantities, with each comparison required for every finite reference $R$. The \emph{completely less-noisy} order $\LNorder$ means
\begin{equation}
 I(U:RE)_{\cM(\rho)}\ge I(U:RB)_{\cN(\rho)}
 \label{eq:cln}
\end{equation}
for every classical label $U$ and ensemble on $RA$ \cite{HircheRouzeFranca2022}. Since $I(U:R)$ is common to both outputs, this is equivalently a comparison of $I(U:E|R)$ and $I(U:B|R)$. Recent work on conditional contraction coefficients identifies the same order with a comparison of conditional relative entropies, allowing the compared states to have unequal reference marginals \cite[Theorem~1 and Section~3.4]{HircheGeorgeNuradhaWilde2026}. These relations are preorders. They induce partial orders after identifying mutually equivalent experiments.

\begin{definition}[Complete variance-dominating signed lift {\cite[Definition~3.3]{ZhuWang2026}}]
A complex-linear Hermiticity-preserving map $\cJ:\cL(B)\to\cL(E)$ is a complete signed lift of $\cN$ through $\cM$ if
\begin{equation}
 \cM^\dagger\circ\cJ=\cN^\dagger,\qquad
 \Var_{\cM(\tau)}(\cJ(Y))\le\Var_{\cN(\tau)}(Y)
 \label{eq:lift}
\end{equation}
for every finite reference $R$, input state $\tau_{RA}$, and Hermitian score $Y_{RB}$.
\end{definition}
The first condition preserves the expectation of every observable at every input state. The variance inequality is therefore equivalent to the corresponding second-moment inequality. Taking $Y=I$ at a faithful output gives $\cJ(I)=I$, so $\cJ^\dagger$ is trace preserving and $\cN=\cJ^\dagger\circ\cM$. Neither map is required to be positive. The distinction from physical postprocessing is that one common map reproduces observable means and controls their fluctuations without necessarily reproducing quantum states.

The next theorem places physical postprocessing, signed lifts, measured $\chi^2$, and the entropic comparisons in one hierarchy.
\begin{resultbox}
\begin{theorem}[Hierarchy of channel orders]\label{thm:hierarchy}
For channels with a common input,
\begin{equation}
\begin{array}{@{}>{\displaystyle}c@{\quad}>{\displaystyle}l@{}}
 \cN=\cD\cM & (\cD\text{ CPTP})\\[3pt]
 \Downarrow & \\[3pt]
 \Sorder & \Longleftrightarrow\quad\text{a complete variance-dominating signed lift exists}\\[3pt]
 \Downarrow & \\[3pt]
 \Morder & \Longrightarrow\quad\Corder
 \quad\Longleftrightarrow\quad\Dorder\quad\Longleftrightarrow\quad\LNorder.
\end{array}
\label{eq:hierarchy}
\end{equation}
To obtain a signed lift from the smoothed order, it suffices to test references with $d_R=d_B^2$.
\end{theorem}
\end{resultbox}
\begin{proof}
For a physical postprocessing $\cD$, the map $\cJ=\cD^\dagger$ is a complete signed lift. Indeed, $\cD^\dagger$ is unital and completely positive, and applying it to the positive block matrix $\left(\begin{smallmatrix}Y^2&Y\\Y&I\end{smallmatrix}\right)$ gives $\cJ(Y)^2\preceq\cJ(Y^2)$ by the Schur complement. Equivalently, $H\mapsto\cD(H)$ preserves the record constraint and decreases measured $\chi^2$, including for Hermitian first arguments. For an arbitrary complete signed lift, unbiasedness and second-moment domination give
\begin{equation}
 2\Tr\cN(X)Y-\Tr\cN(\sigma)Y^2
 \le2\Tr\cM(X)\cJ(Y)-\Tr\cM(\sigma)\cJ(Y)^2.
 \label{eq:pullback}
\end{equation}
for Hermitian $Y$ and $X=\rho-\sigma$. Optimizing over $Y$ proves complete measured ordering. Since $\cJ$ preserves the constraint $Y^{T_R}=Y$, the restricted optimization also proves the smoothed order. For the converse, coupled restricted tests and a minimax argument yield one lift that works for every observable. The details are in \cref{app:lift}.

To prove the measured-to-quantum implication, fix a finite reference $R$, a faithful state $\sigma_{RA}$, and a state $\rho_{RA}$. Apply the coherent-qubit realization of \cref{prop:qubit-realization} to this pair. Complete measured ordering compares its block states $\omega_t,\widehat\sigma_t$ with reference $CR$, so
\begin{equation}
 \mchi\bigl(\cM(\omega_t)\Vert\cM(\widehat\sigma_t)\bigr)
 \ge\mchi\bigl(\cN(\omega_t)\Vert\cN(\widehat\sigma_t)\bigr)
 \quad(0<t<1).
 \label{eq:measured-to-quantum}
\end{equation}
Integrating \eqref{eq:measured-to-quantum} as in \eqref{eq:bkm-mmse} gives $\qchi(\cM(\rho)\Vert\cM(\sigma))\ge\qchi(\cN(\rho)\Vert\cN(\sigma))$. Full-rank approximation extends the inequality to boundary states under the support convention. The quantum-$\chi^2$/relative-entropy equivalence is \cite[Proposition~4.8]{HircheTomamichel2024}, applied for each reference. The relative-entropy/less-noisy equivalence is \cite[Proposition~2.3]{HircheRouzeFranca2022}.
\end{proof}

\Needspace{11\baselineskip}
The complete measured comparison and signed-lift certificate persist under the channel constructions used below.
\begin{resultbox}
\begin{proposition}[Closure of the complete criterion]\label{prop:closure}
For channels $\cN_i$ with a common input and probabilities $p_i$, put $\cN=\sum_i p_i\cN_i$. The following implications hold, with each signed lift taken through the channel's complement:
\begin{equation}
 \begin{aligned}
 \cN_i^c\succeq_{\chi_{\mathrm m}^2}^{\mathrm c}\cN_i\quad(\forall i)
 &\quad\Longrightarrow\quad
 \cN^c\succeq_{\chi_{\mathrm m}^2}^{\mathrm c}\cN,\\
 \cN_i\text{ has a complete signed lift}\quad(\forall i)
 &\quad\Longrightarrow\quad
 \cN\text{ has a complete signed lift}.
 \end{aligned}
 \label{eq:closure-statement}
\end{equation}
Both properties are also preserved by input isometries and output unitaries.
\end{proposition}
\end{resultbox}
\begin{proof}
For $\cN=\sum_i p_i\cN_i$, dephasing a coherent environmental branch label gives $\bigoplus_i p_i\cN_i^c$. Joint convexity, direct sums, and data processing give, with an arbitrary reference retained,
\begin{equation}
 \mchi(\cN\rho\Vert\cN\sigma)
 \le\sum_ip_i\mchi(\cN_i\rho\Vert\cN_i\sigma)
 \le\sum_ip_i\mchi(\cN_i^c\rho\Vert\cN_i^c\sigma)
 \le\mchi(\cN^c\rho\Vert\cN^c\sigma).
 \label{eq:convex-closure}
\end{equation}
Eve can read the branch label, whereas Bob receives an unflagged mixture. A block-diagonal observable lift on the flagged environment preserves means and the second-moment inequality. Restricting the input tests fewer states; output unitaries only change the observable basis. Compressing unused environmental subspaces also preserves the inequality.
\end{proof}

\paragraph{Converse implications.}
The smoothed order tests several scores within one reference-assisted experiment, not merely one score at a time. As shown in \cref{app:lift}, the joint tests yield Hermitian images of an observable basis that preserve all means and obey the second-moment bound for every reference. Extending these images linearly gives the complete signed lift. Separate state-dependent optimal estimators do not suffice. By comparison, reference-assisted guessing-probability orders characterize CPTP postprocessing \cite{Buscemi2016,BuscemiDattaStrelchuk2014}. The Hermitian-smoothed measured-$\chi^2$ order does not require its associated signed lift to preserve positivity of POVM effects.

For unital qubit channels, complete measured ordering has a converse to the hierarchy. If the source Bloch matrix is singular, even an ordinary measured comparison forces physical postprocessing.
\begin{resultbox}
\begin{proposition}[Unital-qubit converse]\label{prop:qubit-converse}
For unital qubit-to-qubit channels with a common input,
\begin{equation}
 \Morder\quad\Longleftrightarrow\quad\Sorder.
\end{equation}
If $\cM$ has a singular Bloch matrix, ordinary measured-$\chi^2$ ordering already implies CPTP postprocessing:
\begin{equation}
 \cM\succeq_{\chi_{\mathrm m}^2}\cN
 \quad\Longrightarrow\quad
 \cN=\cD\circ\cM\qquad(\cD\text{ CPTP}).
\end{equation}
\end{proposition}
\end{resultbox}
\begin{proof}
If the source Bloch matrix is invertible, the source has full output span and the unique unbiased observable map is a complete signed lift by the argument in \cref{app:converses}. Otherwise, measured ordering at the maximally mixed input gives $T_N^{\mathsf T}T_N\preceq T_M^{\mathsf T}T_M$ for the two Bloch matrices. Hence the rule $C(T_Mv)=T_Nv$ defines a contraction on the image of $T_M$, which has dimension at most two. The unused axis lets the input and output rotations be proper and hence unitary. In these bases, extend $C$ to the Bloch matrix $\diag(r,s,rs)$, where $r,s\in[0,1]$ are its singular values (zero if absent). Its four Pauli probabilities are $(1\pm r)(1\pm s)/4\ge0$, so the extension is CPTP and $\cN=\cD\circ\cM$. The reverse implications follow from \cref{thm:hierarchy}.
\end{proof}

We also compare complete orders with their regularized versions. Write $\cM\succeq_D^{\rm reg}\cN$ when the ordinary relative-entropy comparison holds for $\cM^{\otimes n},\cN^{\otimes n}$ for every $n$. Define the measured $\chi^2$, quantum $\chi^2$, and less-noisy versions similarly. Entangled block inputs are allowed, but there is no untouched quantum reference.

The complete measured-$\chi^2$ order tensorizes by replacing one channel factor at a time, treating the other factor's output and any external system as the reference. Thus, for arbitrary joint inputs,
\begin{equation}
 (\cM_1\succeq_{\chi_{\mathrm m}^2}^{\mathrm c}\cN_1,
 \ \cM_2\succeq_{\chi_{\mathrm m}^2}^{\mathrm c}\cN_2)
 \ \Longrightarrow\
 \cM_1\otimes\cM_2\succeq_{\chi_{\mathrm m}^2}^{\mathrm c}\cN_1\otimes\cN_2.
 \label{eq:tensor}
\end{equation}
Iterating this comparison gives the measured-$\chi^2$ implication below. The less-noisy implication follows from complete less-noisy tensorization \cite[Lemma~3.1]{HircheRouzeFranca2022}. In each case, we take the external reference to be trivial:
\begin{equation}
 \begin{aligned}
 \cM\succeq_{\rm ln}^{\mathrm c}\cN
 &\ \Longrightarrow\ \cM\succeq_{\rm ln}^{\rm reg}\cN,\\
 \cM\succeq_{\chi_{\mathrm m}^2}^{\mathrm c}\cN
 &\ \Longrightarrow\ \cM\succeq_{\chi_{\mathrm m}^2}^{\rm reg}\cN.
 \end{aligned}
 \label{eq:complete-to-regularized}
\end{equation}
\Cref{fig:classes} displays the resulting inclusions between complete and regularized channel classes.

\begin{figure}[t]
\centering\includegraphics[width=\textwidth]{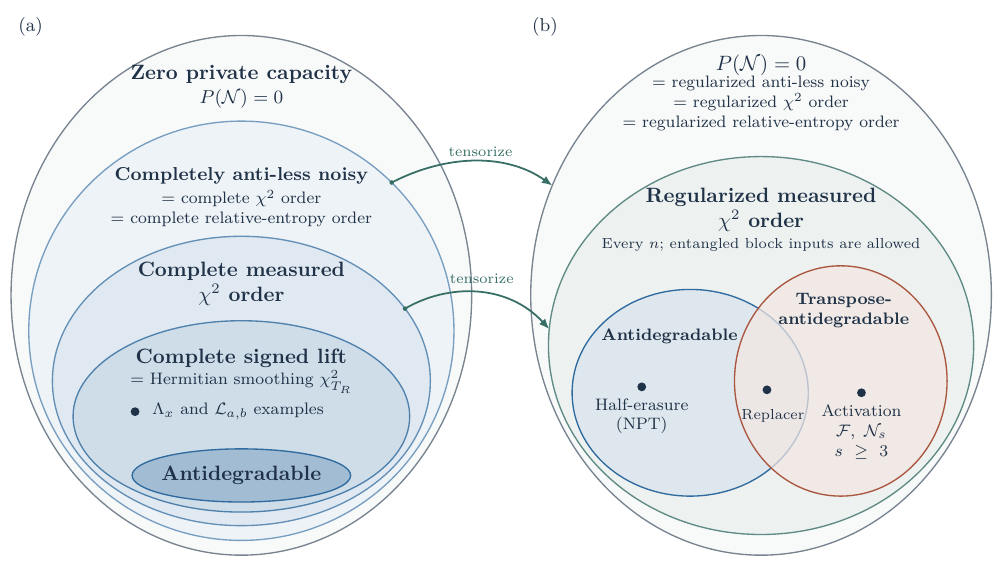}
\caption{Channel classes defined by comparing $\cN^c$ with $\cN$. Nesting denotes inclusion and does not assert strictness of every intermediate inclusion. Panel (a) displays reference-complete criteria. Panel (b) displays criteria imposed at every tensor power without an external reference; its inner region is regularized measured-$\chi^2$ domination, whereas regularized anti-less-noisiness coincides with the outer zero-private-capacity class. The cross-panel arrows are the tensorization implications in \eqref{eq:complete-to-regularized}. The Pauli channel $\mathcal F$ and the coherent-ladder channels $\cN_s$, $s\ge3$, satisfy regularized measured domination but not complete anti-less-noisiness.}
\label{fig:classes}
\end{figure}

\subsection{\texorpdfstring{Quantum $\chi^2$ divergence}{Quantum chi-square divergence}}\label{sec:bkm}
We use the quantum $\chi^2$ divergence of \cite[Section~2.6 and Eq.~(2.62)]{HircheTomamichel2024}, obtained from the quantum $f$-divergence with $f(x)=(x-1)^2$. For a faithful state $\sigma=\sum_i s_i\ket i\bra i$, its quadratic form has both the integral definition and the spectral expression
\begin{equation}
 \begin{aligned}
 \qchi_\sigma(X)&:=\int_0^\infty
 \Tr\!\left[X(\sigma+sI)^{-1}X(\sigma+sI)^{-1}\right]\,ds\\
 &=\sum_{i,j}
 \frac{\ln s_i-\ln s_j}{s_i-s_j}\,
 |\bra iX\ket j|^2,\\
 \qchi(\rho\Vert\sigma)&:=\qchi_\sigma(\rho-\sigma).
 \end{aligned}
 \label{eq:bkm}
\end{equation}
The quotient has the continuous value $1/s_i$ when $s_i=s_j$. For nonfaithful $\sigma$, the integral defines the lower-semicontinuous extension: $\qchi(\rho\Vert\sigma)=+\infty$ unless $\supp\rho\subseteq\supp\sigma$, and otherwise the integral is taken on $\supp\sigma$. This is the Bogoliubov--Kubo--Mori (BKM) quadratic form, equal to $\ln2$ times the Hessian metric of the bit-valued relative entropy \cite{PetzToth1993,LesniewskiRuskai1999}. We use the $\chi^2$ terminology throughout.

The two quadratic forms describe different prediction gains. Reporting a state $\tau$ when the true state is $\omega$ incurs expected logarithmic loss
\begin{equation}
 -\Tr\omega\log_2\tau=S(\omega)+D(\omega\Vert\tau).
\end{equation}
Consider an equally likely classical label $U\in\{+1,-1\}$ encoded in the states $\sigma\pm uX$, where $\Tr X=0$ and $u$ is sufficiently small. The reduction in optimal logarithmic loss when the label is known is the classical--quantum mutual information $I(U:B)$:
\begin{equation}
 I(U:B)=S(\sigma)-\tfrac12S(\sigma+uX)-\tfrac12S(\sigma-uX)
 =\tfrac{u^2}{2\ln2}\qchi_\sigma(X)+O(u^4).
 \label{eq:curvature}
\end{equation}
Thus $\qchi_\sigma$ describes the local advantage in logarithmic prediction \cite{BlumeKohoutHayden2006,Frongillo2022}, using the second-derivative identity of \cite[Theorem~2.8]{HircheTomamichel2024}. For the same ensemble, the squared-loss estimation gain is $u^2q_\sigma(X)$.

\Needspace{13\baselineskip}
The following proposition expresses quantum $\chi^2$ as an integral of measured $\chi^2$ divergences on an enlarged system that includes a coherent qubit $C$. Its MMSE interpretation follows from \cref{prop:mmse}.
\begin{resultbox}
\begin{proposition}[Coherent-qubit realization]\label{prop:qubit-realization}
Let $\rho,\sigma$ be states on $A$ with $\sigma>0$, put $X=\rho-\sigma$, and choose $h>0$ such that $h\norm{\sigma^{-1/2}X\sigma^{-1/2}}_\infty\le1$. Let $C\cong\mathbb C^2$ be a qubit. For $0<t<1$, define the following block operators on $C\otimes A$:
\begin{equation}
 \widehat\sigma_t=\begin{pmatrix}t\sigma&0\\0&(1-t)\sigma\end{pmatrix},\qquad
 \widehat X=\begin{pmatrix}0&X\\X&0\end{pmatrix},\qquad
 \omega_t=\widehat\sigma_t+h\sqrt{t(1-t)}\widehat X.
\end{equation}
Then $\omega_t$ is a state and
\begin{equation}
 \qchi(\rho\Vert\sigma)=\frac1{4h^2}\int_0^1
 \frac{\mchi(\omega_t\Vert\widehat\sigma_t)}{t(1-t)}\,dt.
 \label{eq:bkm-mmse}
\end{equation}
\end{proposition}
\end{resultbox}
\begin{proof}
Set $K=h\sigma^{-1/2}X\sigma^{-1/2}$. Then $\widehat\sigma_t^{-1/2}\omega_t\widehat\sigma_t^{-1/2}=\left(\begin{smallmatrix}I&K\\K&I\end{smallmatrix}\right)$ is positive because $\norm K_\infty\le1$. Also $\Tr\omega_t=1$. Write $G$ in the same qubit-block form as $\widehat\sigma_t$. Since $\widehat\sigma_t$ is block diagonal, the upper-right block of $\cS_{\widehat\sigma_t}(G)=\widehat X$ is
\begin{equation*}
 \begin{aligned}
 [\cS_{\widehat\sigma_t}(G)]_{01}
 &=\tfrac12[(\widehat\sigma_tG)_{01}+(G\widehat\sigma_t)_{01}]\\
 &=\tfrac12[t\sigma G_{01}+(1-t)G_{01}\sigma]
 =(\widehat X)_{01}=X.
 \end{aligned}
\end{equation*}
In an eigenbasis $\sigma=\sum_i s_i\ket i\bra i$, with $X_{ij}=\bra iX\ket j$, its $(i,j)$ entry gives $(G_{01})_{ij}=2X_{ij}/[ts_i+(1-t)s_j]$ and $G_{10}=G_{01}^\dagger$. The two off-diagonal blocks contribute equally to $\Tr\widehat XG$, giving
\begin{equation}
 q_{\widehat\sigma_t}(\widehat X)
 =4\sum_{i,j}\frac{|X_{ij}|^2}{ts_i+(1-t)s_j},\qquad
 \qchi_\sigma(X)=\tfrac14\int_0^1q_{\widehat\sigma_t}(\widehat X)\,dt.
 \label{eq:qubit-bridge}
\end{equation}
In an eigenbasis of $\sigma$, integrating $[ts_i+(1-t)s_j]^{-1}$ gives the logarithmic divided difference in \eqref{eq:bkm} \cite{LesniewskiRuskai1999}. Since $\omega_t-\widehat\sigma_t=h\sqrt{t(1-t)}\widehat X$, the measured-$\chi^2$ identity gives $\mchi(\omega_t\Vert\widehat\sigma_t)=h^2t(1-t)q_{\widehat\sigma_t}(\widehat X)$. Now $\qchi_\sigma(X)=\qchi(\rho\Vert\sigma)$ proves \eqref{eq:bkm-mmse}.
\end{proof}
The block calculation \eqref{eq:qubit-bridge} also holds for arbitrary Hermitian $X$: $\widehat X$ is traceless even if $X$ is not, and a sufficiently small $h$ makes the two block states positive. Thus quantum $\chi^2$ is an integral of measured $\chi^2$ divergences, equivalently of gains in physical squared-loss experiments. The qubit carries the signal in its off-diagonal blocks. Dephasing it would remove the perturbation. Applying a complete measured comparison to each block-state pair and integrating gives the quantum-$\chi^2$ comparison in the proof of \cref{thm:hierarchy}.

\subsection{Receiver dominance and single-letter capacities}\label{subsec:receiver-dominance}
The same comparisons have a capacity consequence in the opposite direction. For $\cN:A\to B$ and its complement $\cN^c:A\to E$, receiver dominance $\cN\succeq_{\mathrm{ln}}^{\mathrm{reg}}\cN^c$ means $I(U:B^n)\ge I(U:E^n)$ for every blocklength and every input ensemble. By the private-capacity formula, this is equivalent to $P(\cN^c)=0$: Watanabe's less-noisy channel class \cite[Definition~2]{watanabe2012private}. Alongside the single-use private information defined in \eqref{eq:intro-private-capacity}, define the single-use coherent information by
\begin{equation}
 Q^{(1)}(\cN):=\max_{\rho_A}\bigl[S(\cN(\rho_A))-S(\cN^c(\rho_A))\bigr].
 \label{eq:single-letter-definitions}
\end{equation}
The complementarity bounds of \cite[Corollary~2.4 and Theorem~2.8]{HircheLeditzky2023} make the role of this premise quantitative. For a channel $\mathcal L$, define
\begin{equation}
 \begin{aligned}
 R_D(\mathcal L):=\sup_{\substack{n\ge1\\\rho,\,\sigma>0}}\frac1n\bigl[
 &D(\mathcal L^{\otimes n}(\rho)\Vert\mathcal L^{\otimes n}(\sigma))\\
 &-D((\mathcal L^c)^{\otimes n}(\rho)\Vert(\mathcal L^c)^{\otimes n}(\sigma))\bigr],
 \end{aligned}
 \label{eq:relative-entropy-advantage}
\end{equation}
where $\rho,\sigma$ are states on the $n$ inputs and $\sigma$ is faithful. This measures the regularized relative-entropy advantage over the complement. The cited bounds imply
\begin{equation}
 \begin{aligned}
 0\le Q(\cN)-Q^{(1)}(\cN)&\le R_D(\cN^c),\\
 0\le P(\cN)-P^{(1)}(\cN)&\le Q(\cN^c)+R_D(\cN^c),\\
 0\le P(\cN)-Q(\cN)&\le Q(\cN^c).
 \end{aligned}
 \label{eq:complementarity-capacity-bounds}
\end{equation}
Thus the complement controls both the gain from regularization and the gap between private and quantum capacity.

\Needspace{20\baselineskip}
These bounds give the following single-letter criterion and its sufficient conditions.
\begin{resultbox}
\begin{proposition}[Receiver dominance and single-letter capacities]\label{prop:single-letter}
Receiver dominance at every blocklength removes both regularizations:
\begin{equation}
 \begin{gathered}
 \cN\succeq_{\mathrm{ln}}^{\mathrm{reg}}\cN^c
 \quad\Longleftrightarrow\quad P(\cN^c)=0\\
 \Longrightarrow\quad
 P(\cN)=P^{(1)}(\cN)=Q(\cN)=Q^{(1)}(\cN).
 \end{gathered}
 \label{eq:receiver-single-letter}
\end{equation}
In particular, the complete criteria supply sufficient conditions:
\begin{equation}
 \begin{gathered}
 \cN\succeq_{\chi_{T_R}^2}^{\mathrm c}\cN^c
 \ \Longrightarrow\
 \cN\succeq_{\chi_{\mathrm m}^2}^{\mathrm c}\cN^c\\
 \Longrightarrow\quad
 P(\cN)=P^{(1)}(\cN)=Q(\cN)=Q^{(1)}(\cN).
 \end{gathered}
 \label{eq:receiver-order-criterion}
\end{equation}
\end{proposition}
\end{resultbox}
\begin{proof}
The regularized less-noisy/relative-entropy equivalence gives $R_D(\cN^c)=0$. Also $Q(\cN^c)\le P(\cN^c)=0$. All three upper bounds in \eqref{eq:complementarity-capacity-bounds} vanish, proving \eqref{eq:receiver-single-letter}. This recovers Watanabe's exact theorem \cite[Theorem~2]{watanabe2012private} as the zero-advantage case of the complementarity bounds. For \eqref{eq:receiver-order-criterion}, use \cref{thm:hierarchy} and tensorize the complete comparison.
\end{proof}
Dominance at every blocklength is essential; a one-use comparison does not suffice. If $Q(\cN^c)=0$, only the last bound in \eqref{eq:complementarity-capacity-bounds} vanishes, giving $P(\cN)=Q(\cN)$ \cite[Theorem~1]{watanabe2012private} without single-letterization. \Cref{sec:examples} gives nondegradable channels with single-letter capacities as complements of zero-private-capacity examples.

\section{Zero private capacity}\label{sec:privacy}
By \eqref{eq:tensor}, complete measured-$\chi^2$ comparison holds at every blocklength. The next theorem turns it into a finite-code bound ruling out positive private capacity.
\begin{resultbox}
\begin{theorem}[Reliability--secrecy bound]\label{thm:privacy}
Suppose the following comparison holds for every $n\ge1$ and every equiprobable binary encoding into $A^n$:
\begin{equation}
 \MMSE(U|E^n)\le\MMSE(U|B^n).
 \label{eq:privacy-estimation}
\end{equation}
For any block code with $M\ge2$ equiprobable messages, average receiver decoding error $\eta\le1/2$, and secrecy error
$\delta=\tfrac12\norm{\rho_{UE}-\tau_U\otimes\rho_E}_1$, where $\tau_U=I_U/M$,
\begin{equation}
 (1-2\eta)^2\le2\delta.
 \label{eq:privacy-bound}\equationalias{eq:intro-reliability-secrecy}
\end{equation}
Consequently $P(\cN)=0$. The estimation comparison holds whenever $\cN^c\succeq_{\chi_{\mathrm m}^2}^{\mathrm c}\cN$.
\end{theorem}
\end{resultbox}
\begin{proof}
By \cref{prop:mmse}, \eqref{eq:privacy-estimation} compares measured $\chi^2$ for two output letters against their midpoints. A complete measured comparison implies this condition at every blocklength by tensorization. For codewords $m\ne m'$, \cref{lem:distance} gives
\begin{equation}
 d_B(m,m')^2\le\mchi(\omega_m\Vert\bar\omega_{m,m'})
 \le\mchi(\zeta_m\Vert\bar\zeta_{m,m'})\le d_E(m,m'),
\end{equation}
where $\omega_m,\zeta_m$ are Bob's and Eve's output states and the bars denote pairwise averages. Let $e_m$ be the conditional decoding error, set $\bar\zeta=\rho_E=M^{-1}\sum_m\zeta_m$, and put $r_m=\tfrac12\norm{\zeta_m-\bar\zeta}_1$. The decoder event $\widehat U=m$ has probabilities $1-e_m$ and at most $e_{m'}$ under messages $m$ and $m'$. This event and the triangle inequality around $\bar\zeta$ give
\begin{equation}
 d_B(m,m')\ge(1-e_m-e_{m'})_+,
 \qquad d_E(m,m')\le r_m+r_{m'}.
\end{equation}
Thus $(1-e_m-e_{m'})_+^2\le r_m+r_{m'}$ for each pair.
Each message index occurs in $M-1$ of the $\binom M2$ unordered distinct pairs. Hence, for any numbers $a_m$,
\[
 \binom M2^{-1}\sum_{m<m'}(a_m+a_{m'})=\frac2M\sum_m a_m.
\]
Equiprobability gives $\eta=M^{-1}\sum_m e_m$. Moreover, $\rho_{UE}-\tau_U\otimes\rho_E=M^{-1}\sum_m\ket{m}\bra{m}\otimes(\zeta_m-\bar\zeta)$ is block diagonal. Its trace norm adds across the message blocks, so $\delta=M^{-1}\sum_m r_m$. The pair-counting identity therefore gives averages $2\eta$ and $2\delta$. Averaging the pairwise inequality and applying Jensen's inequality to the convex function $x\mapsto(1-x)_+^2$ gives $(1-2\eta)^2\le2\delta$, since $\eta\le1/2$. This proves \eqref{eq:privacy-bound}. For a binary code, $\delta=d_E/2$. For any message size, the bound excludes a sequence of reliable secret codes.
\end{proof}
The bound also has a direct estimation interpretation. Bob's binary decision, viewed as an estimate in $\{+1,-1\}$, has squared risk $4\eta$. The comparison gives Eve an estimator with no larger optimal risk, and thresholding it at zero has error at most its squared risk. Eve need not reproduce Bob's state or POVM. Completeness is used to establish the comparison at every blocklength. The communication protocol itself uses no external reference.

Zero private capacity rules out reliable, secret communication at every positive asymptotic rate. The general exponential strong converse \cite[Theorems~2.5 and 2.6]{ChengHircheTomamichel2026} makes this quantitative: at each rate $R>0$, the squared fidelity $F_n^2$ to an ideal secret key is at most $e^{-nE(R)}$ for some $E(R)>0$. Thus the ideal-key error $1-F_n^2$ approaches one exponentially fast. Wilde established an exponential strong converse for degradable channels and at zero for antidegradable channels \cite{WildePrivateStrongConverse2026}. The bound \eqref{eq:privacy-bound} also gives a finite-code tradeoff for separate decoding and trace-distance secrecy errors.

The zero-private-capacity condition also has an exact regularized channel-order form. Combining \eqref{eq:intro-regularized-less-noisy} with the ordinary less-noisy, relative-entropy, and quantum-$\chi^2$ order equivalences at each tensor power \cite[Proposition~2.3]{HircheRouzeFranca2022} and \cite[Proposition~4.8]{HircheTomamichel2024} gives
\begin{equation}
 P(\cN)=0\quad\Longleftrightarrow\quad
 \cN^c\succeq_{\chi^2}^{\rm reg}\cN.
 \label{eq:regularized-zero}
\end{equation}
The binary argument in \cref{thm:privacy} also shows that $\cN^c\succeq_{\chi_{\mathrm m}^2}^{\rm reg}\cN$ implies $P(\cN)=0$.

\begin{figure}[t]
\centering\includegraphics[width=\textwidth]{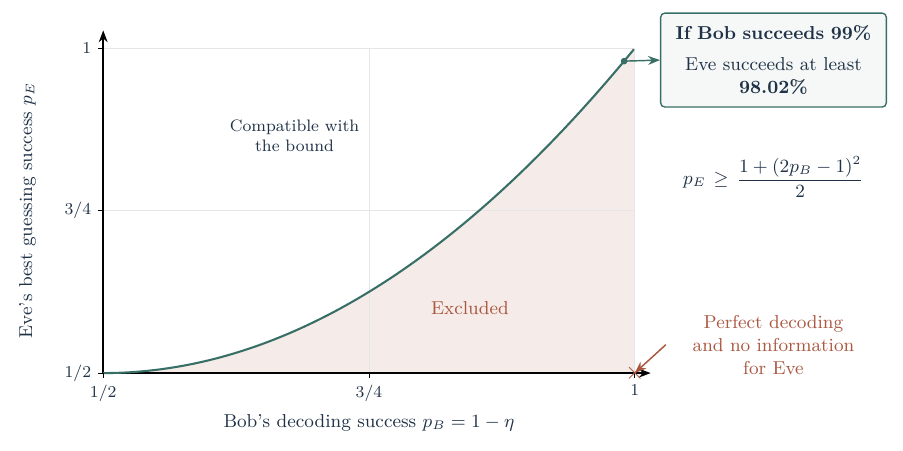}
\caption{Reliable decoding forces information leakage under \cref{thm:privacy}. The plot uses an equiprobable binary code: $p_B=1-\eta$ is Bob's decoding success and $p_E=\tfrac12+\delta$ is Eve's optimal guessing success, with $p_E\ge[1+(2p_B-1)^2]/2$. In $(\eta,\delta)$ coordinates, the same bound holds for every equiprobable message set; only the guessing-probability interpretation is binary. The shaded region, including the marked point of perfect decoding and random guessing, is excluded. For example, $p_B=99\%$ forces $p_E\ge98.02\%$.}
\label{fig:estimation-privacy}
\end{figure}

\Needspace{18\baselineskip}
\subsection{Transpose antidegradability}
Transpose simulation provides a regularized criterion without requiring a physical antidegrading map. The following corollary records the resulting comparisons.
\begin{resultbox}
\begin{corollary}[Transpose antidegradability]\label{cor:transpose}
If a CPTP map $\cD$ satisfies $\cD\cN^c=T\cN$, then
\begin{equation}
 \cN^c\succeq_{\chi_{\mathrm m}^2}^{\rm reg}\cN,\qquad
 \cN^c\succeq_{\chi^2}^{\rm reg}\cN,\qquad
 \cN^c\succeq_D^{\rm reg}\cN.
 \label{eq:transpose-regularized-orders}
\end{equation}
In particular, $P(\cN)=0$.
\end{corollary}
\end{resultbox}
\begin{proof}
The tensorized simulator produces the full transpose on $B^n$. Invariance under full transposition and data processing prove all three divergence comparisons. Equivalently, the observable map $\cJ=\cD^\dagger T$ is positive, unital, and unbiased. Applying the unital-CP inequality from the proof of \cref{thm:hierarchy} to $T(Y)$ gives $\cJ(Y)^2\preceq\cJ(Y^2)$ for Hermitian $Y$, since $T(Y)^2=T(Y^2)$. The same argument applies to its tensor powers. The quantum-$\chi^2$ relation and \eqref{eq:regularized-zero} give $P(\cN)=0$. An untouched quantum reference instead produces a partial transpose, so these arguments do not establish a complete comparison.
\end{proof}
This condition is related to conjugate degradability \cite{BradlerDutilHaydenMuhammad2010}. It implies PPT because $T\cN$ is CP, but PPT alone does not imply zero private capacity \cite{SmithYard2008}. Antidegradability and transpose antidegradability overlap, for example on replacement channels, without either containing the other: a nontrivial half-erasure channel is NPT, while \cref{thm:coherent-family} gives transpose-antidegradable channels that are not antidegradable. The bi-PPT implication to entanglement breaking \cite{MullerHermesSingh2022} concerns a different condition, namely PPT of both complementary channels.

\FloatBarrier

\subsection{Erasure-assisted communication}
Symmetric side-channel assistance is an established way to test what a zero-capacity helper can reveal about a channel \cite[Section~III]{Smith2008Symmetric}; its connection to channel orders and nondegradable single-letter capacities is discussed in \cite[Section~4]{HircheLeditzky2023}. The half-erasure helper also appears in the Smith--Yard construction for quantum-capacity superactivation \cite{SmithYard2008}. The next proposition gives a direct complete-less-noisy formulation of this idea, with a finite helper constructed from any violating ensemble.
\begin{resultbox}
\begin{proposition}[Characterization by erasure assistance]\label{prop:erasure}
For any finite-dimensional channel,
\begin{equation}
 \cN^c\succeq_{\rm ln}^{\mathrm c}\cN
 \ \Longleftrightarrow\ 
 P(\cN\otimes\mathcal E_{d,1/2})=0\quad\text{for every finite }d.
 \label{eq:erasure-test}
\end{equation}
\end{proposition}
\end{resultbox}
\begin{proof}
The forward implication follows by tensorizing the complete comparison with that of the antidegradable erasure channel. Conversely, suppose an ensemble $\{p_x,\rho_x^{AR}\}_{x\in\mathbb Z_m}$ has
$\Delta=I(X:BR)-I(X:ER)>0$. For a uniform message $U\in\mathbb Z_m$, prepare
\begin{equation}
 \sigma_u^{ARK}=\sum_xp_x\rho_x^{AR}\otimes\ket{x-u}\bra{x-u}_K.
 \label{eq:random-label}
\end{equation}
The $AR$ marginal is independent of $u$, so $I(U:B)=I(U:E)=0$.
\Needspace{4\baselineskip}
Unconditionally $K$ is independent of $XAR$; given $K=k$, we have $U=X-k$ and the original distribution of $XAR$. The chain rule cancels the $I(U:K)$ terms:
\begin{equation*}
 \begin{aligned}
 I(U:BRK)-I(U:ERK)
 &=I(U:BR|K)-I(U:ER|K)\\
 &=I(X:BR)-I(X:ER)=\Delta.
 \end{aligned}
\end{equation*}
Here $d_K=m$. Sending $RK$ through $\mathcal E_{d_Rd_K,1/2}$ gives private information
\begin{equation}
 \tfrac12[I(U:BRK)-I(U:E)]
 +\tfrac12[I(U:B)-I(U:ERK)]=\Delta/2>0.
\end{equation}
The private coding theorem therefore gives positive product capacity, proving the converse.
\end{proof}
The proof identifies the reference-assisted information gap with positive private information after half-erasure assistance. The register $K$ is part of the transmitted input and does not require a preshared secret key. When $P(\cN)=0$, the construction gives two zero-private-capacity factors with positive product capacity. Transpose antidegradability supplies a zero-capacity condition without requiring a complete comparison.

The support-asymmetry mechanism discussed by Zhu and Wang \cite[Section~5]{ZhuWangActivation2026} gives a reusable sufficient criterion for activation. A conclusive receiver event produces information proportional to the signal strength $t$, while environmental support containment makes leakage quadratic near $t=0$. Zhu and Wang bound Eve's information by a quadratic expression for domination constants $c\ge3$ and $0\le t\le1/2$ \cite[Lemma~2]{ZhuWangActivation2026}. The bound below instead retains the full binary-outcome information for Bob and the exact classical-coin information that upper-bounds Eve's leakage. It applies to every finite $c\ge1$ and $0\le t\le1$ and is at least as strong on the earlier range.
For states $\rho,\sigma$, we write $D_{\max}(\rho\Vert\sigma)=\ln\inf\{c>0:\rho\preceq c\sigma\}$, with value $+\infty$ if no such $c$ exists.

\Needspace{16\baselineskip}
\begin{resultbox}
\begin{proposition}[Activation from support asymmetry]\label{prop:support-activation}
Let $\cN:A\to B$ have $P(\cN)=0$, and let $\rho_0,\rho_1$ be states on $RA$. Write
\[
 \omega_i=(\id_R\otimes\cN)(\rho_i),\qquad
 \eta_i=(\id_R\otimes\cN^c)(\rho_i),\qquad
 \eta_i^E=\Tr_R\eta_i.
\]
Suppose an effect $0\preceq F\preceq I_{RB}$ satisfies
\begin{equation}
 \Tr F\omega_0=0,\qquad q:=\Tr F\omega_1>0,\qquad
 D_{\max}(\eta_1\Vert\eta_0)<\infty.
 \label{eq:support-activation-hypotheses}
\end{equation}
Set the sharp domination constants
\[
 c_R:=e^{D_{\max}(\eta_1\Vert\eta_0)},\qquad
 c_E:=e^{D_{\max}(\eta_1^E\Vert\eta_0^E)}.
\]
We have $1\le c_E\le c_R<\infty$.
For $c\ge1$ and $0\le t\le1$, set
\begin{equation}
 J_c(t)=h_2\!\left(\frac{1+(c-1)t/4}{c}\right)
 -\frac34h_2(1/c)
 -\frac14h_2\!\left(\frac{1+(c-1)t}{c}\right).
 \label{eq:support-coin-function}
\end{equation}
In particular, $J_1(t)=0$. For $1/2\le p<1$, we have $P(\mathcal E_{d_R,p})=0$, but
\begin{equation}
 P(\cN\otimes\mathcal E_{d_R,p})\ge\max_{0\le t\le1}
 \bigl\{h_2((1-p)qt/4)-\tfrac14h_2((1-p)qt)
 -pJ_{c_R}(t)-(1-p)J_{c_E}(t)\bigr\}>0.
 \label{eq:support-activation-rate}
\end{equation}
\end{proposition}
\end{resultbox}
\begin{proof}
Send $A$ through $\cN$ and $R$ through the erasure helper. Use letters $\rho_0$ and $\rho_t=(1-t)\rho_0+t\rho_1$ with priors $\Pr(U=0)=3/4$ and $\Pr(U=1)=1/4$. If the helper arrives, Bob measures $F$; otherwise he records no click. The conditional click probabilities are $0$ and $(1-p)qt$, so $\Pr(Y=1)=(1-p)qt/4$ and
\[
 I(U:Y)=h_2((1-p)qt/4)-\tfrac34h_2(0)-\tfrac14h_2((1-p)qt).
\]

For a dominated pair $\zeta_1\preceq c\zeta_0$ with $c>1$, the state $\tau=(c\zeta_0-\zeta_1)/(c-1)$ is positive and has trace one. For $0\le z\le1$, define the mixture $\zeta_z=(1-z)\zeta_0+z\zeta_1$. The two letters are $\zeta_0$ and $\zeta_t$, and each can also be prepared as
\[
 \zeta_z=r_z\zeta_1+(1-r_z)\tau,
 \qquad r_z=\frac{1+(c-1)z}{c}.
\]
A classical coin $C$ of bias $r_z$, followed by preparation of $\zeta_1$ or $\tau$ on a system $Z$, produces $\zeta_z$ even when the states do not commute. Since $r_z$ is affine and the priors are $3/4,1/4$, the average coin bias is $r_{t/4}$. Data processing for $U\to C\to Z$ gives the Holevo information bound
\[
 I(U:Z)\le I(U:C)
 =h_2(r_{t/4})-\tfrac34h_2(r_0)-\tfrac14h_2(r_t)
 =J_c(t).
\]
This is the preparation argument of \cite[Lemma~2]{ZhuWangActivation2026}. When $c=1$, the two trace-one states coincide and the information is zero, agreeing with $J_1(t)=0$.

Eve's full output for component $i$ is $p\eta_i\oplus(1-p)\eta_i^E$. The erasure flag is independent of $U$, so her information is at most $pJ_{c_R}(t)+(1-p)J_{c_E}(t)$. The private coding theorem \cite{Devetak2005}, with Bob's measurement fixed, gives the expression in \eqref{eq:support-activation-rate} for every $t$. For finite $c$,
\[
 J_c(t)=\frac{3(c-1)t^2}{32\ln2}+O(t^3),\qquad
 h_2((1-p)qt/4)-\tfrac14h_2((1-p)qt)=\frac{(1-p)qt}{2}+O(t^2).
\]
In the three-entropy expression for $J_c(t)$, the constant and linear Taylor terms cancel, leaving a quadratic term. For Bob, the $t\log t$ terms in the two binary entropies cancel, while $\log_2 4=2$ leaves the stated linear coefficient. Thus the positive term $(1-p)qt/2$ dominates both leakage terms, which are $O(t^2)$, for sufficiently small $t>0$. The helper is antidegradable for $p\ge1/2$, so both factors have zero private capacity.
\end{proof}

The exact rate bound in \eqref{eq:support-activation-rate} yields activation for every sufficiently small signal. The following corollary gives a simpler rate bound valid at each signal strength.

\begin{resultbox}
\begin{corollary}[Quadratic activation rate]\label{cor:support-quadratic}
Under the hypotheses of \cref{prop:support-activation}, define
\[\gamma(c)=(c-1)\max\{1,4/(c+1)\}\qquad(c\ge1)\]
and
\begin{equation}
 \kappa_p=p\gamma(c_R)+(1-p)\gamma(c_E).
 \label{eq:support-activation-parameters}
\end{equation}
For $1/2\le p<1$ and every $0\le t\le1/2$,
\begin{equation}
 P(\cN\otimes\mathcal E_{d_R,p})
 \ge\frac{(1-p)qt}{2}-\frac{3\kappa_pt^2}{32\ln2}.
 \label{eq:support-private-gap}
\end{equation}
\end{corollary}
\end{resultbox}
\begin{proof}
For $c=1$, $J_1(t)=\gamma(1)=0$. For $c>1$, set $f(z)=h_2([1+(c-1)z]/c)$. On $0\le z\le t\le1/2$,
\[
 -f''(z)=\frac{c-1}{\ln2\,[1+(c-1)z](1-z)}
 \le\frac{\gamma(c)}{\ln2}.
\]
The denominator is concave on $[0,1/2]$ and has minimum $\min\{1,(c+1)/4\}$ at an endpoint. Thus $g(z)=f(z)+\gamma(c)z^2/(2\ln2)$ is convex. Jensen's inequality for $g$ at $0,t$ with weights $3/4,1/4$ yields
\begin{equation}
 \begin{aligned}
 J_c(t)&=f(t/4)-\tfrac34f(0)-\tfrac14f(t)\\
 &\le\frac{\gamma(c)}{2\ln2}\left(\frac{t^2}{4}-\frac{t^2}{16}\right)
 =\frac{3\gamma(c)t^2}{32\ln2}.
 \end{aligned}
 \label{eq:coin}
\end{equation}
This is the curvature argument of \cite[Lemma~2]{ZhuWangActivation2026} with the constant extended to $c>1$. Also, a click identifies the second letter and contributes $(1-p)qt/2$ to Bob's information; the no-click contribution is nonnegative.
The click contribution and \eqref{eq:coin} inserted into \eqref{eq:support-activation-rate} prove \eqref{eq:support-private-gap}.
\end{proof}

The register $R$ is a transmitted input to the helper. The $D_{\max}$ condition is equivalent to $\supp\eta_1\subseteq\supp\eta_0$. The criterion therefore applies whenever $\omega_1$ has positive weight outside $\supp\omega_0$: choose $F$ as the projection onto $\ker\omega_0$. In the examples below, the domination constants lie in $\{1\}\cup[3,\infty)$, where $\gamma(c)=c-1$. The Pauli example has $c_E=1$, and its reduced branch leaks no information.

\section{Examples}\label{sec:examples}
The examples distinguish reference-complete comparison from comparisons that hold at every tensor power. \Cref{tab:example-contributions} identifies the cited starting constructions and the results established here for each channel family.

These examples also have a consequence in the opposite direction. Whenever $P(\cN)=0$ and $\cN$ is not antidegradable, its complement $\cN^c$ is nondegradable, yet \cref{prop:single-letter} gives $P(\cN^c)=P^{(1)}(\cN^c)=Q(\cN^c)=Q^{(1)}(\cN^c)$. This is the separation between degradability and regularized less-noisiness identified in \cite[Section~4]{HircheLeditzky2023}. For the Werner--Holevo family, we evaluate the common capacity exactly.

\begin{table}[tbp]
\centering
\small
\renewcommand{\arraystretch}{1.12}
\begin{tabular}{@{}p{.28\textwidth}p{.66\textwidth}@{}}
\toprule
Starting construction & Results established here\\
\midrule
Qutrit channel \cite{ZhuWang2026} & Noisy Werner--Holevo thresholds (\cref{thm:wh}), non-antidegradable mixtures (\cref{ex:mixtures}) and qubit-to-qutrit restrictions (\cref{ex:overlap}), and exact complementary capacities (\cref{cor:wh-complement}).\\
Pauli half-erasure code \cite{privateCapZero2026} & Activation for every $1/2\le p<1$ (\cref{ex:pauli-erasure}) and a finite measured-$\chi^2$ witness (Eq.~\eqref{eq:pauli-diluted-mmse}).\\
Four-level channel, already activating for every $1/2\le p<1$ \cite{ZhuWangActivation2026} & Coherent-ladder extension to every finite $s\ge3$ (\cref{thm:coherent-family}) and reference-assisted witnesses (Eq.~\eqref{eq:mmse-interval}).\\
\bottomrule
\end{tabular}
\caption{Contributions of the channel examples. The cited constructions supply the starting channels or encodings; the right column lists results proved here.}\label{tab:example-contributions}
\end{table}
\FloatBarrier

\Needspace{14\baselineskip}
\subsection{Noisy Werner--Holevo channels}\label{sec:wh}
Our first family extends the qutrit example of Zhu and Wang by varying the Werner--Holevo noise.
Consider the noisy Werner--Holevo channel \cite{RoofehKarimipour2024}
\begin{equation}
 \Lambda_x(X)=(1-x)X+\frac{x}{2}\bigl(\Tr(X)I_3-X^{\mathsf T}\bigr),\qquad 0\le x\le1.
 \label{eq:wh}
\end{equation}
The next theorem determines where this family has zero private capacity and where it is antidegradable.
\begin{examplebox}
\begin{theorem}[Private-capacity and antidegradability thresholds]\label{thm:wh}
For \eqref{eq:wh}, the following conditions are equivalent:
\begin{equation}
 \begin{aligned}
 &\text{a complete signed lift of }\Lambda_x\text{ through }\Lambda_x^c\text{ exists}\\
 &\quad\Longleftrightarrow\quad
 \Lambda_x^c\succeq_{\chi_{\mathrm m}^2}^{\mathrm c}\Lambda_x
 \quad\Longleftrightarrow\quad
 \Lambda_x^c\succeq_{\rm ln}^{\mathrm c}\Lambda_x\\
 &\quad\Longleftrightarrow\quad P(\Lambda_x)=0
 \quad\Longleftrightarrow\quad x\ge\tfrac12.
 \end{aligned}
 \label{eq:wh-complete-threshold}
\end{equation}
Antidegradability holds exactly for $x\ge4/7$. Every member is NPT.
\end{theorem}
\end{examplebox}
\begin{proof}
At $x=1/2$, Zhu and Wang prove a complete variance-dominating signed-lift certificate \cite[Lemma~4.1 and Propositions~3.4 and 4.2]{ZhuWang2026}. The Werner--Holevo endpoint $\Lambda_1$ is self-complementary, as is clear from \eqref{eq:wh-complement}. Thus the closure property in \cref{prop:closure} extends the certificate using
$\Lambda_x=2(1-x)\Lambda_{1/2}+(2x-1)\Lambda_1$ for $x\ge1/2$. This proves the forward claims. For $0<x<1/2$, let $q=x/2$, $b=1-3x/2>q$, $\sigma_\epsilon=\diag(1-2\epsilon,\epsilon,\epsilon)$, and $X=\diag(0,1,-1)$. A direct diagonal calculation gives
\begin{equation}
 q_{\Lambda_x(\sigma_\epsilon)}(\Lambda_x(X))=\frac{2b^2}{q+b\epsilon},\qquad
 q_{\Lambda_x^c(\sigma_\epsilon)}(\Lambda_x^c(X))=\frac{2q}{1-\epsilon}.
 \label{eq:wh-sharp}
\end{equation}
The first is larger for $0<\epsilon<\min\{1/2,(b-q)/b\}$. Both output pairs commute, so measured and quantum $\chi^2$ forms agree. The equiprobable letters $\sigma_\epsilon\pm hX$ are states for $|h|<\epsilon$ and have positive private information for sufficiently small nonzero $h$, by \eqref{eq:curvature}. At $x=0$ the channel is the identity. The symmetric-extension criterion of \cref{app:examples} characterizes antidegradability for this family and reduces to $4(1-x)\le3x$. Thus $\Lambda_x$ is antidegradable exactly for $x\ge4/7$. The partial transpose of the normalized Choi state has eigenvalues $(1-2x)/3$ and $(3x-2)/6$ on two invariant sectors. One is negative for every $x\in[0,1]$.
\end{proof}
Thus $1/2\le x<4/7$ is an explicit interval in which complete signed lifting holds but antidegradability fails.

\Needspace{10\baselineskip}
The complementary channel has a four-dimensional output. The form below uses the environmental basis of \cite[Section~3.2.2, Eq.~(32)]{RoofehKarimipour2024} up to a unitary change of basis. Put
\[
 a(X)=(X_{12}-X_{21},\,X_{20}-X_{02},\,X_{01}-X_{10})^{\mathsf T}.
\]
Then its action on a qutrit matrix $X$ is
\begin{equation}
 \Lambda_x^c(X)=\begin{pmatrix}
 (1-x)\Tr X&\sqrt{x(1-x)/2}\,a(X)^{\mathsf T}\\
 -\sqrt{x(1-x)/2}\,a(X)&(x/2)(\Tr(X)I_3-X^{\mathsf T})
 \end{pmatrix}.
 \label{eq:wh-complement}
\end{equation}
Equivalently, the formula follows from Kraus operators $\sqrt{1-x}I_3$ and $\sqrt{x/2}A_\mu$, where $(A_\mu)_{jk}=\epsilon_{\mu jk}$.

The next corollary evaluates the complementary channel's capacity in the zero-private-capacity range of $\Lambda_x$.
\begin{examplebox}
\begin{corollary}[Capacity of the complementary Werner--Holevo channel]\label{cor:wh-complement}
For $1/2\le x\le1$,
\begin{equation}
 P(\Lambda_x^c)=Q(\Lambda_x^c)=h_2(x)-(1-x)\log_2 3
 \quad\text{bits per use}.
 \label{eq:wh-complement-capacity}
\end{equation}
For $1/2\le x<4/7$, these are positive single-letter capacities of nondegradable channels. At $x=1/2$, the capacity is $1-\tfrac12\log_2 3\simeq0.2075$ bits per use.
\end{corollary}
\end{examplebox}
\begin{proof}
By \cref{thm:wh,prop:single-letter}, $P(\Lambda_x^c)=Q^{(1)}(\Lambda_x^c)$. The coherent information of $\Lambda_x^c$ is concave because its complement has zero private information \cite[Lemma~2.6]{HircheLeditzky2023}. Real orthogonal covariance makes it invariant under $\rho\mapsto O\rho O^{\mathsf T}$. Averaging over this group gives $I_3/3$, so concavity makes that input optimal. Its receiver and environmental spectra for $\Lambda_x$ are $(1/3,1/3,1/3)$ and $(1-x,x/3,x/3,x/3)$, respectively. Reversing their roles gives \eqref{eq:wh-complement-capacity}. Nondegradability of $\Lambda_x^c$ is equivalent to non-antidegradability of $\Lambda_x$, whose threshold is $4/7$.
\end{proof}
The capacity curve is shown in panel~(a) of \cref{fig:example-regions}.

\Needspace{10\baselineskip}
Convex mixing also yields full-rank, non-antidegradable channels near $\Lambda_{1/2}$ while preserving the complete comparison.
\begin{examplebox}
\begin{example}[Full-rank channels without antidegrading maps]\label{ex:mixtures}
Let $\tau_3=I_3/3$. For $0<\epsilon<1/81$, define
\[
 \mathcal M_\epsilon=(1-\epsilon)\Lambda_{1/2}+\epsilon\mathcal R_{\tau_3}.
\]
Every $\mathcal M_\epsilon$ has a full-rank Choi matrix and a complete signed lift through its complement, but is not antidegradable. In particular, $P(\mathcal M_\epsilon)=Q(\mathcal M_\epsilon)=0$.
\end{example}
\end{examplebox}
The replacer component makes the Choi matrix full rank. The symmetric-extension calculation in \cref{app:examples} shows that antidegradability begins only at $\epsilon=1/81$.

\Needspace{17\baselineskip}
\subsection{Qubit-to-qutrit channels}\label{sec:qubitqutrit}
A qubit-input restriction of the qutrit channel $\Lambda_{1/2}$ of Zhu and Wang \cite{ZhuWang2026} leads to the following two-parameter family. Define
\begin{equation}
 \begin{gathered}
 V_0=\ket0\bra0+\ket1\bra1,\qquad V_1=\ket1\bra0-\ket2\bra1,\\
 \mathcal L_{a,b}(X)=aV_0XV_0^\dagger+bV_1XV_1^\dagger+(1-a-b)\Tr(X)\ket0\bra0.
 \end{gathered}
 \label{eq:overlap}
\end{equation}
The channel encodes into the lower pair of levels with probability $a$, into the upper pair with a reversed relative phase with probability $b$, and otherwise resets to $\ket0$. The receiver has no branch flag. The member $\mathcal L_{1/2,1/4}$ is unitarily equivalent to an input restriction of $\Lambda_{1/2}$; mixing it with a half-erasure channel and a replacer generates the certified region below. The other two recent constructions underlie the Pauli and coherent-ladder examples in \cref{sec:pauli,sec:coherent} \cite{privateCapZero2026,ZhuWangActivation2026}.
The next example separates the complete-comparison conclusion from the additional non-antidegradability test.
\begin{examplebox}
\begin{example}[A qubit-to-qutrit separation]\label{ex:overlap}
If
\begin{equation}
 a,b\ge0,\quad a\le2b,\quad a+2b\le1,
 \label{eq:overlap-region}
\end{equation}
then $\mathcal L_{a,b}$ has a complete signed lift through its complement and $P(\mathcal L_{a,b})=Q(\mathcal L_{a,b})=0$. It is NPT whenever $b>0$. If, in addition, $3\sqrt2\,a+4b>3$, then it is non-antidegradable. In particular, $(a,b)=(1/2,1/4)$ satisfies all these conditions.
\end{example}
\end{examplebox}
\begin{proof}
To see the complete comparison, let $u=(\ket0-i\ket1)/\sqrt2$, $v=\ket2$, $U=(u,v)$ and $B=(u,v,\bar u)$. Then $B^\dagger\Lambda_x(UXU^\dagger)B=\mathcal L_{1-x,x/2}(X)$, and
\begin{equation}
 \mathcal L_{a,b}=2a\mathcal L_{1/2,1/4}+(2b-a)\mathcal L_{0,1/2}+(1-a-2b)\mathcal R_{\tau_0}.
\end{equation}
Here $\tau_0=\ket0\bra0$, and the last two channels are, respectively, a half-erasure channel and a replacement channel. The first term inherits a complete signed lift from $\Lambda_{1/2}$; the other two are antidegradable and therefore have physical lifts. The nonnegative coefficients in \eqref{eq:overlap-region} and \cref{prop:closure} give a complete signed lift and zero private and quantum capacity throughout that region. The partial-transpose calculation in \cref{app:examples} proves NPT for $b>0$. The same appendix shows that antidegradability requires $3\sqrt2a+4b\le3$.
\end{proof}

The qubit-to-qutrit parameter region is shown in panel~(b) of \cref{fig:example-regions}.

\begin{figure}[tb]
\centering\includegraphics[width=\textwidth]{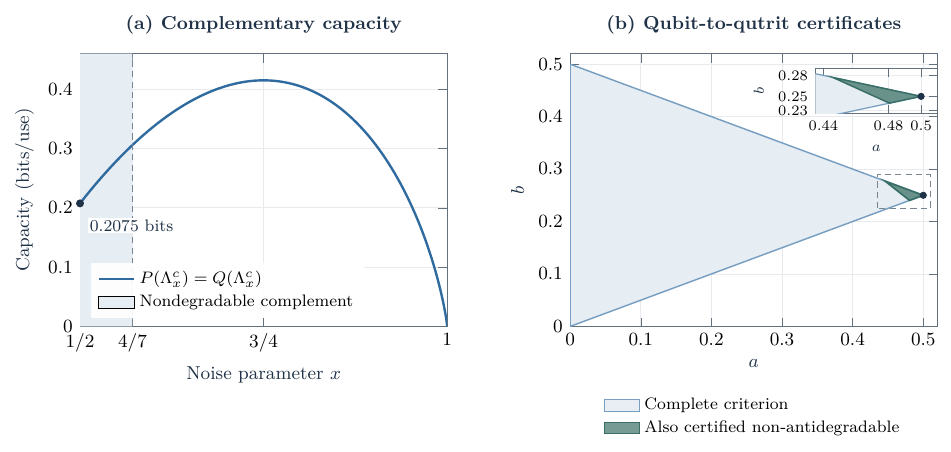}
\caption{Explicit capacity and parameter conclusions. (a) The exact complementary capacity in \eqref{eq:wh-complement-capacity}; the shaded interval $1/2\le x<4/7$ consists of nondegradable complements. The original channel has zero private capacity throughout the plotted interval. (b) Throughout the region \eqref{eq:overlap-region}, $\mathcal L_{a,b}$ has a complete signed lift and zero private and quantum capacity; it is NPT except at the origin. The darker corner also satisfies $3\sqrt2a+4b>3$ and is certified non-antidegradable. The inset enlarges that corner.}
\label{fig:example-regions}
\end{figure}
\FloatBarrier

\Needspace{17\baselineskip}
\subsection{Pauli channels}\label{sec:pauli}
A correlated Pauli channel and its binary encoding were introduced by Uthirakalyani G, Halder, and Elkouss \cite{privateCapZero2026}. They proved activation with an antidegradable Pauli helper and with a qubit erasure helper at $p=1/2$. We retain their channel and letters, derive reference-assisted measured-$\chi^2$ comparisons, and certify a positive private rate for every $1/2\le p<1$. The all-$p$ erasure result of Zhu and Wang \cite{ZhuWangActivation2026} concerns a different four-level channel, treated in \cref{sec:coherent}.

For two-qubit Pauli operators, write $AB=A\otimes B$. Their correlated Pauli channel is
\begin{equation}
 \mathcal F(X)=\frac16\sum_{P\in\mathcal P}PXP^\dagger,
 \qquad\mathcal P=\{XX,YI,IZ,XZ,IY,YZ\}.
 \label{eq:pauli-channel}
\end{equation}
The next example bounds the private rate for this fixed Pauli channel across the full erasure range.
\begin{examplebox}
\begin{example}[Pauli erasure assistance]\label{ex:pauli-erasure}
For $1/2\le p<1$, $P(\mathcal F)=P(\mathcal E_{2,p})=0$, but
\begin{equation}
 P(\mathcal F\otimes\mathcal E_{2,p})
 \ge\frac{\ln2\,(1-p)^2}{216p(1+\sqrt3)}>0
 \quad\text{bits per product use}.
 \label{eq:pauli-erasure-rate}
\end{equation}
Thus the fixed Pauli channel activates every qubit erasure channel with erasure probability below one and zero private capacity.
\end{example}
\end{examplebox}

\Needspace{12\baselineskip}
\subsubsection{Channel and encoding}
This is a uniform mixture of correlated local errors. Let $\{\ket P:P\in\mathcal P\}$ be orthonormal environmental labels for the six errors. The matrix elements of a complementary channel and a transpose-antidegrading map are
\begin{equation}
 \bra P\mathcal F^c(X)\ket Q=\frac16\Tr(QPX),\qquad
 \cD(\ket P\bra Q)=\frac14(QP)^{\mathsf T},\qquad
 \cD\mathcal F^c=T\mathcal F.
 \label{eq:pauli-simulator}
\end{equation}
Complete positivity and trace preservation of $\cD$ follow from the isometry $W\ket P=(I_4\otimes P)\ket{\Phi_4}$, followed by tracing its second subsystem, where $\ket{\Phi_4}=\frac12\sum_{j=0}^3\ket{jj}$. The six Bell vectors are orthonormal. Thus \cref{cor:transpose} gives all tensor-power comparisons and $P(\mathcal F)=0$. The full complement retains coherences between the error labels.

Let $K,L,C$ be qubits. The operators $\Phi^\pm$ and $\Psi^\pm$ project onto $(\ket{00}\pm\ket{11})/\sqrt2$ and $(\ket{01}\pm\ket{10})/\sqrt2$, respectively, on $LC$. The binary letters are
\begin{align}
 \sigma_0^{KLC}&=\tfrac12\ket0\bra0_K\otimes\Phi^+_{LC}
 +\tfrac14\ket1\bra1_K\otimes(\Psi^-_{LC}+\Phi^-_{LC}),\notag\\
 \sigma_1^{KLC}&=\tfrac14\ket0\bra0_K\otimes(\Phi^+_{LC}+\Psi^+_{LC})
 +\tfrac12\ket1\bra1_K\otimes\Phi^-_{LC}.
 \label{eq:pauli-letters}
\end{align}
Both have $KL$ marginal $I_4/4$ and $C$ marginal $I_2/2$. Write
\begin{equation}
 \omega_u=(\mathcal F\otimes\id_C)(\sigma_u),\qquad
 \eta_u=(\mathcal F^c\otimes\id_C)(\sigma_u).
 \label{eq:pauli-outputs}
\end{equation}
Here $B$ and $E$ denote the receiver and environmental outputs of $\mathcal F$. Their marginals are independent of $u$, so the message is carried entirely by correlations with $C$. For uniform priors, the output calculation in \cref{app:pauli} gives
\begin{align}
 I_B:=I(U:BC)&=\tfrac14\log_2(4/3),\notag\\
 I_E:=I(U:EC)&=\tfrac12\left[h_2\!\left(\frac{3-\sqrt3}{6}\right)
 -h_2\!\left(\frac{3-\sqrt5}{6}\right)\right].
 \label{eq:pauli-information}
\end{align}
Sending $C$ through a half-erasure channel gives the rate of \cite{privateCapZero2026},
\begin{equation}
 P(\mathcal F\otimes\mathcal E_{2,1/2})\ge R_0:=\tfrac12(I_B-I_E)
 \simeq0.003390>0.
 \label{eq:pauli-half-rate}
\end{equation}
The antidegradable two-qubit Pauli helper $\mathcal G$ of \cite{privateCapZero2026} gives the same rate with a fixed input flag. The positive product rate already proves that $\mathcal F$ is not antidegradable and its complement is not completely less noisy: either condition would force zero capacity with an antidegradable helper. By the hierarchy, complete measured domination and a complete signed lift also fail.

\subsubsection{Reference-assisted estimation}
For the original binary ensemble, Bob has more mutual information with the label $U$ than Eve by \eqref{eq:pauli-information}, although Eve estimates $U$ more accurately in mean square. With $\bar\omega=(\omega_0+\omega_1)/2$ and $\bar\eta=(\eta_0+\eta_1)/2$, exact SLD evaluation yields
\begin{equation}
 \mchi(\omega_1\Vert\bar\omega)=\frac19,
 \qquad\mchi(\eta_1\Vert\bar\eta)=\frac4{33}.
 \label{eq:pauli-original-mmse}
\end{equation}
In this reference-assisted experiment with an equiprobable label in $\{+1,-1\}$, Bob's MMSE is $8/9$, whereas Eve's is $29/33<8/9$. Thus the original pair does not witness failure of complete measured-$\chi^2$ ordering: for that, we need an ensemble on which Bob has the smaller MMSE.

Retain equal priors and replace the second letter by $\sigma_{1/2}=(\sigma_0+\sigma_1)/2$. More generally, write $\sigma_t=(1-t)\sigma_0+t\sigma_1$ and use subscripts $t$ for its outputs. With $\bar\omega_t=(\omega_0+\omega_t)/2$ and $\bar\eta_t=(\eta_0+\eta_t)/2$, we obtain
\begin{equation}
 \mchi(\omega_{1/2}\Vert\bar\omega_{1/2})=\frac{11}{315}
 >\frac{140}{4257}=\mchi(\eta_{1/2}\Vert\bar\eta_{1/2}).
 \label{eq:pauli-diluted-mmse}
\end{equation}
This gives Bob the smaller MMSE with only a qubit reference. The formulas in \eqref{eq:pauli-line-chi} cover every dilution parameter; \cref{fig:pauli-dilution} shows where the comparison changes sign. At the support boundary the distinction is stronger: $\mchi(\omega_1\Vert\omega_0)=+\infty$, while $\mchi(\eta_1\Vert\eta_0)=32/45$. Equation~\eqref{eq:pauli-diluted-mmse} avoids divergent values.

\begin{figure}[tb]
\centering\includegraphics[width=.86\textwidth]{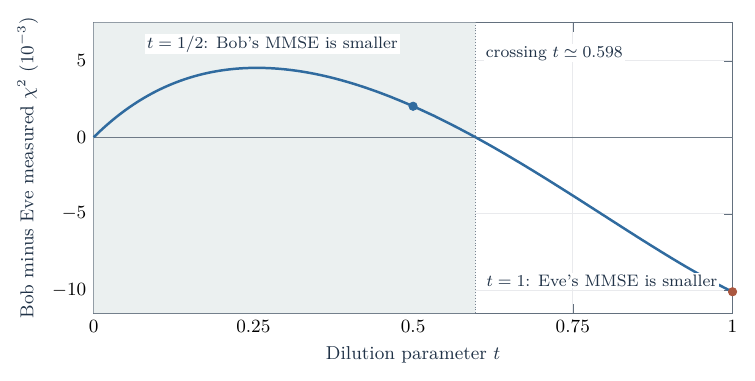}
\caption{Dilution reverses the reference-assisted estimation comparison for the Pauli channel. The curve is $\mchi(\omega_t\Vert\bar\omega_t)-\mchi(\eta_t\Vert\bar\eta_t)$ from \eqref{eq:pauli-line-chi}; positive values mean Bob has smaller MMSE for this binary ensemble. The shaded region shows that advantage. The witness at $t=1/2$ is positive, whereas the original pair at $t=1$ has a negative difference. The dotted sign change is at $t\simeq0.598$.}
\label{fig:pauli-dilution}
\end{figure}
\FloatBarrier

We now turn from the estimation witness to the assisted communication rate. For the unchanged equal-prior code, erasure probability $p$ gives private information $(1-p)I_B-pI_E$, which is positive only for $p<I_B/(I_B+I_E)\simeq0.5169$. Dilution extends the positive rate throughout $1/2\le p<1$.
\begin{proof}[Proof of \cref{ex:pauli-erasure}]
The projector onto $\ket{1,\psi^-}$ has probability zero under $\omega_0$ and $q=1/12$ under $\omega_1$. The environmental states obey
\begin{equation}
 \eta_1\preceq(2+\sqrt3)\eta_0.
 \label{eq:pauli-support}
\end{equation}
The generalized eigenvalues on their support are $1$ and $2\pm\sqrt3$, as shown in \cref{app:pauli}. Their reduced environmental states coincide. Thus \cref{prop:support-activation} applies with $c_R=2+\sqrt3$, $c_E=1$, $q=1/12$, and $\kappa_p=p(1+\sqrt3)$. The signal choice
$t_*=2\ln2(1-p)/[9p(1+\sqrt3)]<1/2$
inserted in \cref{cor:support-quadratic} gives \eqref{eq:pauli-erasure-rate}. The exact bound \eqref{eq:support-activation-rate} is stronger: at $p=1/2$, taking $t=3/50$ gives a rate greater than $6.117\times10^{-4}$ bits per product use.
\end{proof}
Dilution gives a positive rate throughout $1/2\le p<1$ for the fixed channel. At half erasure, the original encoding gives the stronger rate \eqref{eq:pauli-half-rate}.
\FloatBarrier

\Needspace{10\baselineskip}
\subsection{Coherent ladder channels}\label{sec:coherent}
We now construct a continuous family of four-dimensional channels containing the channel of Zhu and Wang \cite{ZhuWangActivation2026}.

\subsubsection{Error-detecting subspaces}
We first choose a reference-assisted code that detects the nonidentity errors used in the channel construction.
Let $A=\mathbb C^4$ have the ordered basis $\ket0,\ket1,\ket2,\ket3$, and let $R=\mathbb C^2$. For $s>0$, choose
\begin{equation}
 u_s=\frac{\sqrt s\ket{00}+\ket{11}}{\sqrt{s+1}},\qquad
 v_s=\frac{\ket{02}+\sqrt s\ket{13}}{\sqrt{s+1}},\qquad
 w=\frac{\ket{01}+\ket{12}}{\sqrt2},
 \label{eq:code}
\end{equation}
and let $P_s$ project onto their three-dimensional span. We determine the local error operators satisfying
\begin{equation}
 P_s(I_R\otimes L)P_s=0.
 \label{eq:detect}
\end{equation}
Condition~\eqref{eq:detect} maps the code subspace into its orthogonal complement and hence permits error detection \cite{KnillLaflamme1997}.

The operators $L\in\mathcal L(A)$ satisfying \eqref{eq:detect} are exactly the linear combinations of
\begin{align}
 D_s&=\diag(1,-s,s,-1),&
 A_s&=\ket0\bra1-\sqrt s\ket1\bra2+\ket2\bra3,\\
 B&=\ket0\bra2-\ket1\bra3,& C&=\ket0\bra3,
 \label{eq:errors}
\end{align}
together with $A_s^\dagger,B^\dagger,C^\dagger$. Here $L_{jk}=\bra jL\ket k$ denotes a matrix entry of an arbitrary operator $L$ satisfying \eqref{eq:detect}. The diagonal constraints give $sL_{00}+L_{11}=L_{11}+L_{22}=L_{22}+sL_{33}=0$. The first upper diagonal obeys $\sqrt s L_{01}+L_{12}=L_{12}+\sqrt s L_{23}=0$, the second obeys $L_{02}+L_{13}=0$, and the third is free; the lower diagonals give adjoints. Together, the nine independent constraints on sixteen entries determine the seven-dimensional error space.

In particular,
\begin{equation}
 \bra{u_s}(I\otimes A_s)\ket w
 =\frac{\sqrt s-\sqrt s}{\sqrt{2(s+1)}}=0.
 \label{eq:cancel}
\end{equation}
The three transitions in $A_s$ share a single Kraus label, so the two amplitudes in \eqref{eq:cancel} interfere. Resolving them into distinct environmental labels would instead add their probabilities and remove the cancellation. The parameter $s$ is the central-to-outer intensity ratio for one-step transitions; the Schmidt amplitudes of the code vectors are varied accordingly.

\subsubsection{Channel family and dilation}
The channel combines an identity branch with the coherent errors detected by the preceding code.
Set $\nu_s=2s(s+3)$ and define
\begin{equation}
\begin{aligned}
 \cN_s(X)=\frac1{\nu_s}\bigl[&2sX+D_sXD_s
 +s(A_sXA_s^\dagger+A_s^\dagger XA_s)\\
 &+3s(BXB^\dagger+B^\dagger XB)
 +(2s^2-1)(CXC^\dagger+C^\dagger XC)\bigr].
\end{aligned}
\label{eq:coherent-family}
\end{equation}
For $s\ge1/\sqrt2$, all Kraus weights are nonnegative, and direct multiplication gives the normalized identity
\[
 \begin{aligned}
 \frac1{\nu_s}\bigl[2sI+D_s^\dagger D_s
 &+s(A_s^\dagger A_s+A_sA_s^\dagger)
 +3s(B^\dagger B+BB^\dagger)\\
 &+(2s^2-1)(C^\dagger C+CC^\dagger)\bigr]=I.
 \end{aligned}
\]
Thus \eqref{eq:coherent-family} is CPTP. These weights also satisfy the transpose-simulation condition proved below. Pairing each jump with its adjoint makes the channel unital and Hilbert--Schmidt self-adjoint. The family contains the Zhu--Wang channel at
\begin{equation}
 \cN_4=\mathcal Z.
 \label{eq:Z}
\end{equation}
For $s\ge3$, the input and receiver output dimensions are $d_A=d_B=4$, and the minimal environmental dimension is $d_E=8$.

Write
\begin{equation}
 (L_e)=(I,D_s,A_s,A_s^\dagger,B,B^\dagger,C,C^\dagger),\quad
 (r_e)=(2s,1,s,s,3s,3s,2s^2-1,2s^2-1),
 \label{eq:raw-errors}
\end{equation}
so $V_s\ket\psi=\nu_s^{-1/2}\sum_e\sqrt{r_e}L_e\ket\psi\otimes\ket e$ is a dilation. Take
\begin{equation}
 H_A=\epsilon\diag(0,1,2,3),\qquad H_E=\epsilon\diag(3,3,4,2,5,1,6,0).
\end{equation}
The identity $(H_A\otimes I+I\otimes H_E)V_s=V_s(H_A+3\epsilon I)$ shows that the dilation preserves total energy when the environment is initialized in $\ket0_E$. Extending the isometry $V_s$ to a unitary in each energy sector gives a finite collision model. Each environmental label groups transitions with the same energy exchange, preserving their coherent sum. The privacy analysis uses the full environmental output, including coherences between labels. The identity branch has input-independent probability $1/(s+3)$.

The next theorem gives a transpose-antidegradability criterion for this family and a private-rate bound after erasure assistance.
\begin{examplebox}
\begin{theorem}[Coherent-ladder family]\label{thm:coherent-family}
On the CPTP domain of \eqref{eq:coherent-family},
\begin{equation}
 \cN_s\text{ transpose-antidegradable}
 \ \Longleftrightarrow\ 
 \cN_s\text{ PPT}\ \Longleftrightarrow\ s\ge3.
 \label{eq:coherent-range}
\end{equation}
For every finite $s\ge3$ and erasure probability $1/2\le p<1$, put
\begin{equation}
 c_s=\frac{(s+1)(s^2+6s+1)}{(s-1)^2},\qquad
 \kappa_{s,p}=p(c_s-1)+(1-p)(s-1).
 \label{eq:coherent-constants}
\end{equation}
Then $P(\cN_s)=P(\mathcal E_{2,p})=0$, but
\begin{equation}
 P(\cN_s\otimes\mathcal E_{2,p})
 \ge\frac{2\ln2\,(1-p)^2}{3(s+3)^2\kappa_{s,p}}>0
 \quad\text{bits per product use}.
 \label{eq:coherent-rate}
\end{equation}
Consequently $\cN_s$ is not antidegradable. Its complement is regularized less noisy and regularized measured-$\chi^2$ dominant, but not completely less noisy or completely measured-$\chi^2$ dominant.
\end{theorem}
\end{examplebox}

\Needspace{8\baselineskip}
\begin{proof}
\emph{Transpose antidegradability.} \Cref{app:simulator} constructs a CPTP map $\cD_s$ with $\cD_s\cN_s^c=T\cN_s$ for $s\ge3$: its Choi positivity follows from the rank-one sum \eqref{eq:simulator-gram}, and \eqref{eq:simulator-identities} verifies trace preservation and composition. Thus \cref{cor:transpose} gives $P(\cN_s)=0$. The same appendix computes an eigenvalue $2s(s-3)/\nu_s$ of the partially transposed Choi matrix and proves positivity for $s\ge3$. Since transpose antidegradability implies PPT, \eqref{eq:coherent-range} follows.

\medskip
\noindent\emph{Activation.} Apply \cref{prop:support-activation} to the background and signal states
\begin{equation}
 \rho_0^{(s)}=\tfrac12(u_su_s^\dagger+v_sv_s^\dagger),\qquad
 \rho_1=ww^\dagger,
 \label{eq:background}
\end{equation}
Both reference marginals equal $I_R/2$, so the helper input alone contains no information about the letter.

By \eqref{eq:detect}, every state $\rho$ in the code obeys
\begin{equation}
 P_s(\id_R\otimes\cN_s)(\rho)P_s=\frac{\rho}{s+3}.
 \label{eq:accepted}
\end{equation}
Only the identity branch survives this joint code projection. Consequently $F=ww^\dagger$ has probability zero on the background and $q=1/(s+3)$ on the signal.

For $\eta_i=(\id_R\otimes\cN_s^c)(\rho_i)$ and $\eta_i^E=\Tr_R\eta_i$, \cref{app:environment} proves the domination bounds
\begin{equation}
 \eta_1\preceq c_s\eta_0,\qquad\eta_1^E\preceq s\eta_0^E.
 \label{eq:env-bounds}
\end{equation}
The proposition therefore applies with $c_R=c_s$, $c_E=s$, $q=1/(s+3)$, and $\kappa_p=\kappa_{s,p}$. Since $(1-p)/(s+3)\le1/12$ and $\kappa_{s,p}\ge2$, the choice $t_*=8(1-p)\ln2/[3(s+3)\kappa_{s,p}]$ lies below $1/2$. Substitution in \cref{cor:support-quadratic} proves \eqref{eq:coherent-rate}; the helper has zero private capacity because it is antidegradable. A complete less-noisy comparison would force zero capacity with this helper, contradicting the positive rate. The hierarchy also excludes antidegradability and complete measured domination. The regularized comparisons follow from \cref{cor:transpose}.
\end{proof}

At half erasure the bound becomes
\begin{equation}
 P(\cN_s\otimes\mathcal E_{2,1/2})
 \ge\frac{\ln2\,(s-1)^2}{3(s+3)^2(2s^3+3s^2+12s-1)}.
 \label{eq:half-rate}
\end{equation}
At $s=4$ the bound recovers the original $3\ln2/10927$ bits \cite{ZhuWangActivation2026}; at $s=3$ it is $\ln2/3132\simeq2.213\times10^{-4}$ bits per product use. The exact rate bound \eqref{eq:support-activation-rate} improves the latter guarantee: at $s=3$, $p=1/2$, and $t=1/80$, it exceeds $2.475\times10^{-4}$ bits per product use.
Panel~(a) of \cref{fig:activation-bounds} plots the closed-form rate bound \eqref{eq:coherent-rate} as $s$ varies for three erasure probabilities.

\begin{figure}[tb]
\centering\includegraphics[width=\textwidth]{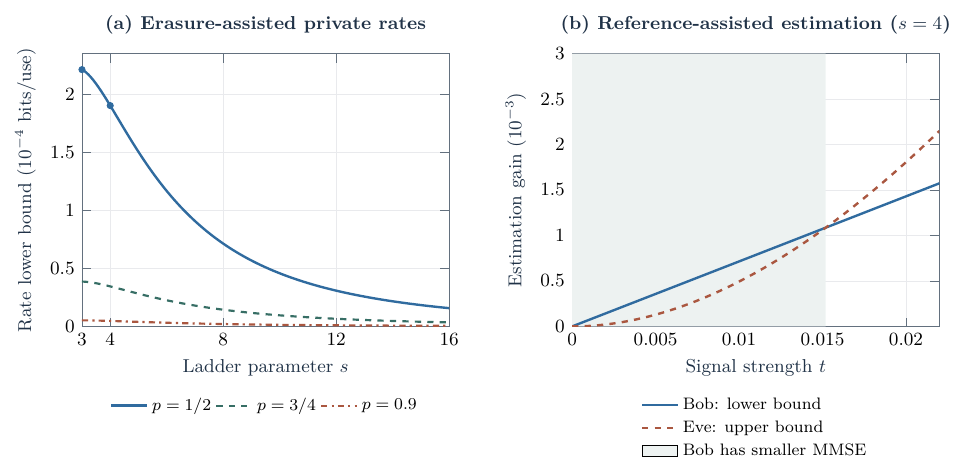}
\caption{Numerical values of proved ladder bounds. (a) Achievable private-rate lower bounds from \eqref{eq:coherent-rate}, for three erasure probabilities; the individual channels have zero private capacity. (b) For $s=4$, the receiver's lower bound and the upper bound for Eve in \eqref{eq:mmse-witness}. The shaded interval is \eqref{eq:mmse-interval}, throughout which Bob has smaller MMSE in the same reference-assisted experiment.}
\label{fig:activation-bounds}
\end{figure}

\paragraph{Reference-assisted estimation.}
The failure of complete measured ordering can also be observed directly. Use equal priors for $\rho_0^{(s)}$ and $\rho_t^{(s)}=(1-t)\rho_0^{(s)}+t\rho_1$, let $\omega_i=(\id_R\otimes\cN_s)(\rho_i)$, and use bars to denote the respective average output states. The click measurement and the preparation coin give
\begin{equation}
 \mchi(\omega_t\Vert\bar\omega)\ge\frac{t/(s+3)}{2-t/(s+3)},\qquad
 \mchi(\eta_t\Vert\bar\eta)\le\frac{(c_s-1)t^2}{(2-t)[2+(c_s-1)t]}.
 \label{eq:mmse-witness}
\end{equation}
\Needspace{14\baselineskip}
The first inequality is the classical $\chi^2$ value of the click measurement. For the second, the preparation in the proof of \cref{prop:support-activation} maps coins of bias $1/c_s$ and $[1+(c_s-1)t]/c_s$ to $\eta_0$ and $\eta_t$. The measured-$\chi^2$ value between the second coin and their average is the displayed upper bound, so data processing applies even for noncommuting outputs. This measured-$\chi^2$ estimate is sharp when only the domination constant is specified: orthogonal prepared states attain it. Comparing the two rational expressions shows that Bob has smaller MMSE whenever
\begin{equation}
 0<t<\frac{2}{1+(s+2)(c_s-1)}.
 \label{eq:mmse-interval}
\end{equation}
Thus an entire explicit interval witnesses failure of the complete comparison, although Eve dominates every unassisted blocklength comparison. Panel~(b) of \cref{fig:activation-bounds} shows the two bounds and this interval at $s=4$.

\Needspace{10\baselineskip}
\section{Discussion}\label{sec:discussion}
Our results place several zero-private-capacity criteria in one hierarchy. Fixed-record Hermitian smoothing characterizes complete variance-dominating signed lifts. These imply complete measured-$\chi^2$ ordering, which in turn implies complete quantum-$\chi^2$, relative-entropy, and less-noisy ordering. Binary estimation yields a quantitative reliability--secrecy bound; the coherent-qubit representation connects that bound to the established entropic approach.

The examples distinguish physical simulation from reference-stable information comparison. The noisy Werner--Holevo channels and their qubit-to-qutrit restrictions have zero private capacity without antidegradability, yet their complete comparison prevents activation by an antidegradable helper. The Pauli construction and the coherent-ladder family instead satisfy comparisons at every tensor power through transpose simulation, while a joint measurement with a transmitted reference yields positive private information. A conclusive receiver event gives a linear signal gain, whereas the support-contained environmental perturbation contributes only quadratically. Both examples therefore separate regularized from complete less-noisiness with four-dimensional channel input and output. The ladder family adds a continuous deformation in which the transition amplitudes and the error-detecting code preserve destructive interference together.

The Pauli encoding also distinguishes two information comparisons for a fixed ensemble. For the original input states, Bob has more mutual information with the label, whereas Eve achieves smaller MMSE. Diluting one input state gives the finite measured-$\chi^2$ reversal in \eqref{eq:pauli-diluted-mmse}. The support calculation then extends erasure-assisted activation from the published half-erasure rate to all $1/2\le p<1$. These consequences connect the correlated Pauli noise model to the estimation criterion without identifying finite-ensemble mutual-information and MMSE comparisons.

Whether complete unsmoothed measured-$\chi^2$ ordering implies a complete signed lift remains open in general. \Cref{prop:qubit-converse} settles unital qubit pairs, and \cref{app:converses} also treats sources with full output span. The examples studied here do not settle the remaining case: the qutrit-based families satisfy both conditions and the activating examples fail both. A related question in \cite{privateCapZero2026} asks whether every transpose-antidegradable channel that is not antidegradable activates with a suitable antidegradable helper. By \cref{prop:erasure} and tensorization of the complete comparison, this is equivalent to the unresolved implication
\begin{equation}
 \left.\begin{gathered}
 \cN\text{ transpose-antidegradable},\\
 \cN^c\succeq_{\rm ln}^{\mathrm c}\cN
 \end{gathered}\right\}
 \stackrel{?}{\Longrightarrow}\ \cN\text{ antidegradable}.
\end{equation}
Failure of the complete order supplies a finite-dimensional half-erasure activator, while satisfaction of that order rules out every antidegradable helper. Neither activating example resolves this implication, since both fail the complete order. Further questions concern activation outside the ladder family's certified range, optimized achievable rates, and other physical noise models with the required interference and support properties.

\Needspace{8\baselineskip}
\section*{Acknowledgments}
CH received funding by the Deutsche Forschungsgemeinschaft (DFG, German Research Foundation) -- 550206990, the Federal Ministry of Research, Technology and Space (BMFTR), Germany, under the QC service center QUICS (grant no.~13N17418), and by Germany's Excellence Strategy -- EXC 2123/2 QuantumFrontiers -- 390837967.

\paragraph{AI use.}
OpenAI Codex (models \texttt{gpt-6-sol} and \texttt{gpt-6-astra}, used in September 2026) contributed to the mathematical proofs in this work, including proposing proof strategies, deriving intermediate results, and drafting and revising proof arguments. It also assisted with literature searches, numerical checks, exposition, figures, and \LaTeX{} production. The author retains full responsibility for the validity of the proofs and the accuracy and conclusions of the manuscript.

\clearpage
\appendix
\crefalias{section}{appendix}
\addtocontents{toc}{\protect\setcounter{tocdepth}{0}}
\section{Estimation and channel-order proofs}\label{app:comparison-proofs}
\addtocontents{toc}{\protect\setcounter{tocdepth}{2}}
\addtocontents{toc}{\protect\addvspace{4pt}}
\addtocontents{toc}{\protect\contentsline{section}{Appendices}{\thepage}{\csname @currentHref\endcsname}}
\addtocontents{toc}{\protect\setcounter{tocdepth}{0}}
We prove the converse in \cref{thm:hierarchy} and the sufficient converse conditions discussed in \cref{sec:discussion}. The output-support convention of Section~1 applies throughout.

\subsection{Complete signed-lift characterization}\label[appendix]{app:lift}
We extract coupled quadratic tests and use minimax to choose one unbiased observable map whose second-moment bound holds for all scores. This proves the converse in \cref{thm:hierarchy} without selecting a different estimator for each input state.

\paragraph{Coupled quadratic tests.}
For Hermitian $X_1,\ldots,X_k$ and a positive block matrix $\Omega=[\Omega_{ij}]\succ0$ on $\mathbb C^k\otimes A$, define
\begin{equation}
 Q_\Omega(\bm X)=\sup_{Y_i=Y_i^\dagger}
 \left\{2\sum_i\Tr X_iY_i-\sum_{ij}\Tr\Omega_{ij}Y_jY_i\right\}.
 \label{eq:coupled}
\end{equation}
The off-diagonal blocks of $\Omega$ couple different scores. The quadratic cost is positive definite because it is the expectation of the square of a Hermitian score against a faithful block state. Channels act on $A$ and leave the block index unchanged. Complete restricted ordering implies
\begin{equation}
 Q_{\cN(\Omega)}((\cN(X_i))_i)
 \le Q_{\cM(\Omega)}((\cM(X_i))_i).
 \label{eq:coupled-order}
\end{equation}
To see this, use the reference $\mathbb C\ket0\oplus\mathbb C^k$ and the denominator and direction
\begin{equation}
 \widehat\sigma_\epsilon=\begin{pmatrix}\epsilon\tau&0\\0&\Omega\end{pmatrix},\qquad
 \widehat X=\sum_i(\ket0\bra i+\ket i\bra0)\otimes X_i,
 \label{eq:star}
\end{equation}
with $\tau$ faithful and $\epsilon>0$. Normalize the denominator and choose a sufficiently small state perturbation in the traceless direction $\widehat X$. The normalization is common to both channel outputs and cancels in the comparison.

A score fixed by $T_R$ has Hermitian off-diagonal blocks $H_i=Y_{0i}=Y_{i0}$. Only these blocks contribute to the linear term. The remaining blocks have independent nonnegative quadratic costs, so an optimal score sets them to zero. Substituting $H_i=2Y_i$ leaves
$4Q_{\Phi(\Omega+\epsilon I\otimes\tau)}((\Phi(X_i))_i)$ at output $\Phi$, before the common normalization factor. The quadratic cost remains positive definite as $\epsilon\downarrow0$, so continuity proves \eqref{eq:coupled-order}.

\paragraph{Matching all observable means.}
If $d_B=1$, the unital map on scalars is already a complete lift, so assume $d_B\ge2$. Choose a traceless Hermitian basis $F_1,\ldots,F_k$ of $B$, where $k=d_B^2-1$, and put $G_i=\cN^\dagger(F_i)$. The equations $\cM^\dagger(H_i)=G_i$ are feasible. Indeed, the comparison with a trivial reference is ordinary measured ordering; small perturbations of a faithful input show that $\ker\cM\subseteq\ker\cN$ on Hermitian operators. Taking orthogonal complements gives $\operatorname{im}\cN^\dagger\subseteq\operatorname{im}\cM^\dagger$.

For a positive block weight $\Omega$, define the coupled second-moment cost
\begin{equation}
 C_\Phi(\bm H,\Omega)
 :=\sum_{i,j}\Tr\Omega_{ij}\Phi^\dagger(H_jH_i).
 \label{eq:coupled-cost}
\end{equation}
Indeed, $C_\Phi(\bm H,\Omega)=\Tr S^\dagger(\id\otimes\Phi)(\Omega)S$ for $S=\sum_i\ket i\otimes H_i$, so the cost is nonnegative and convex in the Hermitian tuple $\bm H$. Let $\mathcal A$ be the affine set of tuples satisfying $\cM^\dagger(H_i)=G_i$.
\Needspace{4\baselineskip}
For Hermitian multipliers $X_i$, the Lagrangian is $C_\cM(\bm H,\Omega)+2\sum_i\Tr X_i[G_i-\cM^\dagger(H_i)]$. Minimizing over unrestricted Hermitian $\bm H$ gives the dual objective below by \eqref{eq:coupled}; feasibility and positive definiteness give strong duality:
\begin{align}
 v_\cM(\Omega)
 &:=\inf_{\bm H\in\mathcal A}C_\cM(\bm H,\Omega)\notag\\
 &=\sup_{X_i=X_i^\dagger}\left\{2\sum_i\Tr X_iG_i
 -Q_{\cM(\Omega)}((\cM(X_i))_i)\right\}\notag\\
 &\le C_\cN(\bm F,\Omega).
 \label{eq:quadratic-dual}
\end{align}
The inequality uses \eqref{eq:coupled-order} and the feasible target tuple $F_i$. Thus every positive definite weight admits an unbiased source tuple whose weighted cost is no greater than the target's. The same inequality holds for positive semidefinite weights: apply it to $\Omega+\epsilon I$ and let $\epsilon\downarrow0$.

\paragraph{Choosing one tuple for every weight.}
Let $\mathcal W$ be the compact convex set of positive block weights with trace one. The function
\begin{equation}
 f(\bm H,\Omega)=C_\cM(\bm H,\Omega)-C_\cN(\bm F,\Omega)
\end{equation}
is continuous, convex in $\bm H\in\mathcal A$, and affine in $\Omega\in\mathcal W$. Sion's minimax theorem \cite{Sion1958} and \eqref{eq:quadratic-dual} therefore give
\begin{equation}
 \inf_{\bm H\in\mathcal A}\max_{\Omega\in\mathcal W}f(\bm H,\Omega)
 =\max_{\Omega\in\mathcal W}\inf_{\bm H\in\mathcal A}f(\bm H,\Omega)\le0.
 \label{eq:lift-minimax}
\end{equation}
The maximum includes $\Omega=I/(kd_A)$, so the left-hand objective is coercive on $\mathcal A$: its quadratic part contains a positive multiple of $\sum_i\Tr\cM(I)H_i^2$. The output-support convention makes $\cM(I)$ faithful, and a minimizing tuple exists. Choose one with $f(\bm H,\Omega)\le0$ for every $\Omega\in\mathcal W$.

Set $\cJ(F_i)=H_i$ and $\cJ(I)=I$, then extend complex-linearly. To check the complete bound, take any reference $R$, state $\tau_{RA}$, and Hermitian $Y=C_0\otimes I+\sum_i C_i\otimes F_i$. The blocks
\begin{equation}
 \Omega_{ij}=\Tr_R\!\bigl[(C_i\otimes I)\tau(C_j\otimes I)\bigr]
 \label{eq:lift-weight}
\end{equation}
form a positive weight: they are the blocks of the partial trace over $R$ of $V\tau V^\dagger$, where $V=\sum_i\ket i\otimes C_i\otimes I$. Normalize its trace if nonzero. By \eqref{eq:coupled-cost}, $f(\bm H,\Omega)\le0$ is exactly the second-moment comparison for the nonidentity part of $Y$. The terms containing $C_0$ cancel because $\cM^\dagger\cJ=\cN^\dagger$. Since this also preserves means, the variance bound in \eqref{eq:lift} follows. The tests in \eqref{eq:star} used a reference of dimension $k+1=d_B^2$, as claimed in \cref{thm:hierarchy}.

\subsection{Full-output-span converse}\label[appendix]{app:converses}
We prove the full-output-span step used in \cref{prop:qubit-converse}.

Assume $\operatorname{im}\cM=\cL(E)$ and complete measured ordering. The perturbation and kernel argument of \cref{app:lift} gives $\ker\cM\subseteq\ker\cN$ and, for faithful $\sigma$,
\begin{equation}
 q_{\cN(\sigma)}(\cN(X))\le q_{\cM(\sigma)}(\cM(X))
 \qquad(X=X^\dagger).
\end{equation}
For nontraceless $X$, subtract $(\Tr X)\sigma$; the removed term adds the same $(\Tr X)^2$ to both forms. Surjectivity makes $\cM^\dagger$ injective, so kernel inclusion yields a unique Hermiticity-preserving map $\cJ$ with $\cM^\dagger\cJ=\cN^\dagger$ and $\cJ(I)=I$.

For a target score $Y$, the variational formula \eqref{eq:quadratic} and the displayed comparison imply
$2\Tr\cM(X)\cJ(Y)-q_{\cM(\sigma)}(\cM(X))\le\Tr\cN(\sigma)Y^2$.
Every Hermitian output $H$ has the form $\cM(X)$, so optimizing over $H$ gives
\begin{equation}
 \sup_{H=H^\dagger}\{2\Tr H\cJ(Y)-q_{\cM(\sigma)}(H)\}
 =\Tr\cM(\sigma)\cJ(Y)^2\le\Tr\cN(\sigma)Y^2.
\end{equation}
The same argument applies after adjoining any reference. The amplified source remains surjective, so its unique unbiased map is $\id_R\otimes\cJ$. Thus one map satisfies the second-moment bound for every reference and input; unbiasedness turns it into the required variance bound.

\section{Symmetric-extension certificates}\label{app:nonantidegradable}
The criteria below prove the antidegradability threshold in \cref{thm:wh} and the separation claims in \cref{ex:mixtures,ex:overlap}.

\subsection{Symmetric extensions}\label[appendix]{app:examples}
Antidegradability is equivalent to a symmetric extension of the normalized Choi state \cite{MyhrLutkenhaus2009}. We use this equivalence for the Werner--Holevo threshold, the full-rank mixtures, and the qubit-to-qutrit obstruction.
\subsubsection{Invariant qutrit Choi states}
Let $F$ be the swap on two qutrits and put
\[
 \Phi=\ket\Omega\bra\Omega/3,\quad \ket\Omega=\sum_{i=0}^2\ket{ii},\quad
 P^-=(I-F)/2,\quad P^{\rm st}=(I+F)/2-\Phi.
\]
For $\omega=p_0\Phi+(p_1/3)P^-+(p_2/5)P^{\rm st}$, antidegradability is equivalent to
\begin{equation}
 4p_0\le3p_1+5p_2+2\sqrt{15p_1p_2}.
 \label{eq:brauer}
\end{equation}
We derive the criterion below by specializing the Brauer-state extension geometry \cite{SolymosJakabZimboras2024}.
Let $W=z_0\Phi+z_1P^-+z_2P^{\rm st}$ and define
$\mathfrak H(W)=(W_{AB_1}\otimes I_{B_2}+W_{AB_2}\otimes I_{B_1})/2$.
If $\mathfrak H(W)\succeq0$, every symmetric extension has nonnegative expectation of $W$ on its marginal. Conversely, averaging a separating witness over the orthogonal symmetry of the Choi state gives a witness of this form. The cone of extendible marginals is closed because trace bounds every positive extension. Thus these witnesses characterize extendibility, and their expectation on $\omega$ is simply $p_0z_0+p_1z_1+p_2z_2$.

To make the blocks explicit, for $v\in\mathbb C^3$ define trace tensors
\[
 (a_v)_{ijk}=\delta_{ij}v_k,\qquad
 (b_v)_{ijk}=\delta_{ik}v_j,\qquad
 (c_v)_{ijk}=\delta_{jk}v_i.
\]
The maps
\[
 e_0(v)=\frac{c_v}{\sqrt3},\qquad
 e_+(v)=\frac{3a_v+3b_v-2c_v}{2\sqrt{15}},\qquad
 e_-(v)=\frac{a_v-b_v}{2}
\]
give three orthogonal isometric copies of $\mathbb C^3$. Receiver exchange is even on the first two and odd on the third. The multiplicity blocks of $\mathfrak H(W)$ are the scalars $z_1,z_2$, $(z_1+3z_2)/4$, $(3z_1+z_2)/4$, $(4z_0+3z_1+5z_2)/12$, and the block
\begin{equation}
 E(z_0,z_1,z_2)=
 \begin{pmatrix}
 (z_0+3z_1+5z_2)/9&\sqrt5(2z_0-3z_1+z_2)/18\\
 \sqrt5(2z_0-3z_1+z_2)/18&(20z_0+15z_1+z_2)/36
 \end{pmatrix}.
\end{equation}
The first four blocks act on the trace-free antisymmetric, symmetric, and two mixed-symmetry sectors, of dimensions $1,7,5,5$. There $\mathfrak H(\Phi)=0$ and the averaged swap has eigenvalues $-1,1,\pm1/2$, giving the stated scalars. On the trace tensors, contraction gives the fifth scalar on $e_-$ and $E$ on $(e_0,e_+)$. These spaces exhaust all $27$ dimensions and specialize the invariant extension analysis of \cite{ZhuWang2026,SolymosJakabZimboras2024}.

The first two scalar blocks force $z_1,z_2\ge0$. Scale $z_2=1$ and write $z_1=u$. The lowest permitted value of $z_0$ is exposed by the rank-one factorization
\begin{equation}
 E\!\left(-\frac{4u}{5u+3},u,1\right)
 =\frac{v_uv_u^{\mathsf T}}{12(5u+3)},\qquad
 v_u=\begin{pmatrix}2\sqrt5(u+1)\\1-5u\end{pmatrix}.
\end{equation}
At this value, the remaining scalar is $(5u^2+6u+5)/[4(5u+3)]>0$. Increasing $z_0$ adds a positive multiple of $\mathfrak H(\Phi)$. Decreasing it gives negative expectation in the null direction $(5u-1,2\sqrt5(u+1))^{\mathsf T}$, since the coefficient of $z_0$ in $E$ is $(1,\sqrt5)^{\mathsf T}(1,\sqrt5)/9$. Thus positivity holds exactly above this boundary. Since the marginal expectation increases with $z_0$, it suffices to test
\begin{equation}
 z_2=1,\quad z_1=u\ge0,\quad z_0=-\frac{4u}{5u+3},
 \label{eq:witness-family}
\end{equation}
The cases $z_2=0$ follow by continuity; when $z_1=z_2=0$, positivity simply requires $z_0\ge0$. Multiplying the boundary expectation by $5u+3$ gives
\begin{equation}
 5p_1u^2+(3p_1+5p_2-4p_0)u+3p_2\ge0
 \qquad\text{for every }u\ge0.
\end{equation}
For $A,C\ge0$, the quadratic $Au^2+Bu+C$ is nonnegative on $u\ge0$ exactly when $B\ge-2\sqrt{AC}$. Applying this with $A=5p_1$ and $C=3p_2$ gives \eqref{eq:brauer}, including zero weights by continuity.

\paragraph{The full-rank mixture.}
The normalized Choi state of $\mathcal M_\epsilon$ in \cref{ex:mixtures} has weights
\[
 p_0=\frac{9-7\epsilon}{18},\qquad
 p_1=\frac{3-\epsilon}{6},\qquad
 p_2=\frac{5\epsilon}{9}.
\]
For $0\le\epsilon\le1/81$, the difference $4p_0-3p_1-5p_2=(3-23\epsilon)/6$ is positive, and
\[
 (4p_0-3p_1-5p_2)^2-60p_1p_2
 =\frac{(1-\epsilon)(1-81\epsilon)}4.
\]
Hence \eqref{eq:brauer} fails for $0\le\epsilon<1/81$ and holds at $\epsilon=1/81$; convexity of symmetric extendibility and antidegradability of $\mathcal R_{\tau_3}$ give the threshold for larger $\epsilon$. The replacer term makes the Choi state positive definite whenever $\epsilon>0$. Both $\Lambda_{1/2}$ and $\mathcal R_{\tau_3}$ have complete signed lifts, so \cref{prop:closure} gives one for their mixtures.

\subsubsection{Qubit-to-qutrit channels}
We use the support of the Choi state to obtain the antidegradability obstruction for \eqref{eq:overlap}.
For \eqref{eq:overlap}, put $c=1-a-b$. The normalized Choi state has diagonal $(1-b,b,0,c,a,b)/2$ in the basis $(00,01,02,10,11,12)$, with off-diagonal entries $J_{00,11}=a/2$, $J_{01,12}=-b/2$. Its kernel contains $\ket{02}$ and $\ket{01}+\ket{12}$. Imposing these kernel constraints on both $AB_1$ and $AB_2$ restricts an extension to the following eight-dimensional support:
\begin{align*}
 &h=\ket{100},\quad e_0=\ket{000},\quad e_1=(\ket{101}+\ket{110})/\sqrt2,
 \quad o_0=(\ket{101}-\ket{110})/\sqrt2,\\
 &f_0=(\ket{001}+\ket{010}-\ket{102}-\ket{120})/2,\quad f_1=\ket{111},\\
 &o_1=(\ket{001}-\ket{010}-\ket{102}+\ket{120})/2,
 \quad g=(\ket{011}-\ket{112}-\ket{121})/\sqrt3.
\end{align*}
The Choi state is invariant under $U_A(\theta)=\diag(1,e^{i\theta})$ and $U_B(\theta)=\diag(1,e^{-i\theta},e^{-2i\theta})$. Averaging an extension over $U_A(\theta)\otimes U_B(\theta)\otimes U_B(\theta)$ and receiver exchange preserves positivity and both marginals. It separates the displayed basis by phase and swap parity. Only $(e_0,e_1)$ and $(f_0,f_1)$ can then have nonscalar blocks, written $\left(\begin{smallmatrix}A&C\\\bar C&B\end{smallmatrix}\right)$ and $\left(\begin{smallmatrix}D&F\\\bar F&E\end{smallmatrix}\right)$. Let $O_0,O_1,H,G\ge0$ be the scalar weights. Define $z=(D+O_1)/4$, $v=(B+O_0)/2$, $w=G/3$. The marginal equations are
\begin{equation}
 z+w=b/2,\quad A=(1-b)/2-z,\quad H+v+z=c/2,\quad
 E=a/2-v-w,\quad C/\sqrt2+F/2=a/2.
\end{equation}
Positivity of the two $2\times2$ blocks gives $a/2\le\sqrt{Av}+\sqrt{zE}$; the nonnegative scalar weight $H$ gives $v+z\le c/2$. The following weighted arithmetic--geometric mean bounds are chosen to give the same coefficient to $v$ and $z$:
\begin{equation}
 \sqrt{Av}\le\tfrac{\sqrt2-1}{2}A+\tfrac{\sqrt2+1}{2}v,\qquad
 \sqrt{zE}\le\tfrac{\sqrt2}{2}z+\tfrac1{2\sqrt2}E.
\end{equation}
Substitute the marginal equations and use $v+z\le c/2$ to obtain
\begin{equation}
 \frac a2\le\frac{3\sqrt2-2a-4\sqrt2b}{8}
 \quad\Longrightarrow\quad 3\sqrt2a+4b\le3.
\end{equation}
The strict reverse of this necessary condition gives the non-antidegradability claim in \cref{ex:overlap}. For $b>0$, the partial transpose has a zero diagonal entry coupled to another basis vector by $-b/2\ne0$. A positive semidefinite matrix cannot have a nonzero entry in a row with zero diagonal, by Cauchy--Schwarz. Thus $\mathcal L_{a,b}$ is NPT whenever $b>0$.

\section{Calculations for the activation examples}\label{app:activation-calculations}
We verify the Pauli calculations from \cref{sec:pauli} and the ladder certificates used in \cref{thm:coherent-family}.

\subsection{Pauli-channel calculations}\label[appendix]{app:pauli}
We establish the information values, MMSE reversal, and sharp support bound cited in \cref{sec:pauli}, using the input letters and dilation of \cite{privateCapZero2026}. Non-antidegradability already follows from the positive private rate with an antidegradable helper.

\subsubsection{Output reduction}
We first diagonalize the receiver states and reduce the complementary states to three-dimensional blocks.
In the ordered computational/Bell basis
\begin{equation}
 (\ket{0,\phi^+},\ket{1,\phi^-},\ket{0,\phi^-},\ket{1,\phi^+},
 \ket{0,\psi^+},\ket{1,\psi^-},\ket{0,\psi^-},\ket{1,\psi^+}),
\end{equation}
the receiver states are diagonal, with
\begin{equation}
 12\omega_0=\diag(1,2,2,2,1,0,2,2),\qquad
 12\omega_1=\diag(2,1,2,2,0,1,2,2).
 \label{eq:pauli-bob-blocks}
\end{equation}
For $P=A\otimes Q\in\mathcal P$, the unitary
$U_{EC}=\sum_{P}\ket P\bra P\otimes Q_C^*$ gives
$U_{EC}\eta_iU_{EC}^\dagger=\zeta_i\otimes I_C/2$.
Define
\begin{align}
 a_\pm&=(\ket{IZ}\pm\ket{IY})/\sqrt2,&
 s&=(\ket{XX}+i\ket{YI})/\sqrt2,\notag\\
 p&=(\ket{XZ}+i\ket{YZ})/\sqrt2,&
 m&=(\ket{XZ}-i\ket{YZ})/\sqrt2.
\end{align}
In the basis $(a_+,a_-,s,p,m,s_\perp)$, where
$s_\perp=(\ket{XX}-i\ket{YI})/\sqrt2$, the reduced states are
\begin{equation}
 \zeta_i=\frac{I_2}{6}\oplus A_i\oplus0,\qquad
 A_0=\frac16\begin{pmatrix}2&1&0\\1&1&0\\0&0&1\end{pmatrix},\quad
 A_1=\frac16\begin{pmatrix}2&0&-1\\0&1&0\\-1&0&1\end{pmatrix}.
 \label{eq:pauli-eve-blocks}
\end{equation}
The eigenvalues $1/6,1/6,1/6,0$ are common to both letters and their average, so their entropy contributions cancel in the Holevo information. The two remaining eigenvalues are $(3\pm\sqrt5)/12$ for each letter and $(3\pm\sqrt3)/12$ for the average. Their total weight is $1/2$, giving the environmental expression in \eqref{eq:pauli-information}. The receiver expression follows directly from the diagonal probabilities in \eqref{eq:pauli-bob-blocks}.

\subsubsection{Estimation and support bounds}
Set $A_t=(1-t)A_0+tA_1$. The common $I_2/6$ block contributes zero to measured $\chi^2$, and tensoring with $I_C/2$ leaves it unchanged. A single calculation gives all the estimation values used in the text: for $0\le t,\theta\le1$,
\begin{align}
 \mchi(\eta_t\Vert\eta_\theta)
 &=\frac{16(t-\theta)^2(2+\theta-\theta^2)}
 {9(1+2\theta-2\theta^2)(5+2\theta-2\theta^2)},\notag\\
 \mchi(\omega_t\Vert\bar\omega_t)
 &=\frac{t(2+2t-t^2)}{3(4-t^2)(4-t)}.
 \label{eq:pauli-line-chi}
\end{align}
To obtain the first formula, solve $(A_\theta G+GA_\theta)/2=A_t-A_\theta$. Put $r=1-\theta$, $z=-\theta$, $u=G_{01}$ and $v=G_{02}$. Eliminating the other four entries of the real symmetric score reduces the SLD equation to
\begin{equation}
 \begin{pmatrix}
 3-(3r^2+z^2)/2&-rz\\
 -rz&3-(r^2+3z^2)/2
 \end{pmatrix}\begin{pmatrix}u\\v\end{pmatrix}
 =-2(t-\theta)\begin{pmatrix}1\\1\end{pmatrix}.
\end{equation}
The quadratic value is $\Tr[(A_t-A_\theta)G]=-(t-\theta)(u+v)/3$, giving the first expression in \eqref{eq:pauli-line-chi}. The second follows by summing the classical quadratic form for the diagonal probabilities in \eqref{eq:pauli-bob-blocks}. Taking $\theta=t/2$ and then $t=1$ or $1/2$ proves \eqref{eq:pauli-original-mmse} and \eqref{eq:pauli-diluted-mmse}. Taking $(t,\theta)=(1,0)$ gives the boundary value $32/45$. The zero receiver coordinate in \eqref{eq:pauli-bob-blocks} gives the corresponding infinite value.

For the support bound, $A_0$ and $A_1$ are positive definite, have the same determinant $1/216$, and satisfy $\Tr(A_0^{-1}A_1)=5$. Since $A_1-A_0$ has rank two, one generalized eigenvalue is $1$. The other two have product one and sum four, hence are $2\pm\sqrt3$. The common block adds only eigenvalues $1$. Thus $2+\sqrt3$ is the smallest constant in \eqref{eq:pauli-support}.

\subsection{Coherent ladder channels}\label[appendix]{app:ladder-calculations}
We prove the transpose-simulator and PPT claims in \cref{thm:coherent-family}, followed by the environmental domination bounds in \eqref{eq:env-bounds}.

\subsubsection{Transpose-antidegrading maps}\label[appendix]{app:simulator}
Choi matrices here are unnormalized, with the input system first. At $s=4$, the simulator recovers the certificate of \cite{ZhuWangActivation2026}.

\paragraph{PPT condition.}
The partial transpose of the Choi matrix of $\nu_s\cN_s$ is block diagonal by the sum of the input and output level indices. Reflection pairs sums $0,6$, $1,5$ and $2,4$. Representative first blocks are $(2s+1)$ and $s\left(\begin{smallmatrix}1&1\\1&1\end{smallmatrix}\right)$. The sum-two block, split by reversal parity, is
\begin{equation}
 \begin{pmatrix}6s&-\sqrt2s^{3/2}\\-\sqrt2s^{3/2}&s^2+2s\end{pmatrix}\oplus(0),
\end{equation}
This block has positive diagonal and determinant $4s^2(s+3)$ for $s>0$. The sum-three block is
\begin{equation}
 \begin{pmatrix}2s^2+2s-2&-2s\\-2s&2s\end{pmatrix}
 \oplus
 \begin{pmatrix}2s^2-2s&-4s\\-4s&2s^2-2s\end{pmatrix}.
\end{equation}
The first block has positive diagonal and determinant $4s(s^2-1)$ for $s\ge1$. The last has eigenvalues $2s(s+1)$ and $2s(s-3)$. Thus all blocks are positive exactly when $s\ge3$ in the CPTP domain.

\paragraph{Construction of the simulator.}
In the basis $\ket{e,c}$ of $E\otimes B=\mathbb C^8\otimes\mathbb C^4$, set
\begin{align*}
 z_1&=\tfrac12\ket{0,0}-\tfrac1{2s}\ket{1,0}+\ket{3,1},&
 z_2&=\tfrac12\ket{0,0}+\tfrac1{2s}\ket{1,0}+\ket{5,2},\\
 z_3&=\ket{7,3},&
 z_4&=2\ket{0,2}+\ket{1,2}-\sqrt s\ket{2,1}+\ket{3,3}+3\ket{4,0},\\
 z_5&=\ket{3,0}-\sqrt s\ket{5,1}+\ket{7,2},&
 z_6&=-\ket{5,0}+\ket{7,1},\qquad z_7=\ket{7,0}.
\end{align*}
Write $Z_s=(z_1,z_2,z_3)$ and set
\[
 f_s=\begin{pmatrix}s\\s-2\\s+3+1/s\end{pmatrix},\qquad
 g_s=\begin{pmatrix}0\\1\\(s+1)^2/[2s(s-1)]\end{pmatrix}.
\]
Define the positive operator directly as a sum of rank-one terms:
\begin{equation}
 \begin{aligned}
 Q_s=\frac1{2(s+1)}\biggl[&(Z_sf_s)(Z_sf_s)^\dagger
 +4(s^2-1)(Z_sg_s)(Z_sg_s)^\dagger\\
 &+\frac{4s(s+1)^2(s-3)}{s-1}z_3z_3^\dagger
 +s\sum_{j=4}^6z_jz_j^\dagger
 +(2s^2+4s+3)z_7z_7^\dagger\biggr].
 \end{aligned}
 \label{eq:simulator-gram}
\end{equation}
Every coefficient is nonnegative for $s\ge3$, and the $z_3z_3^\dagger$ coefficient vanishes at $s=3$. Thus the construction includes the endpoint directly.
Let $\mathsf R\ket{e,c}=\varepsilon_e\ket{\pi(e),3-c}$, with
$\pi=(0,1,3,2,5,4,7,6)$ and $\varepsilon=(1,-1,1,1,-1,-1,1,1)$.
Set $H_s=Q_s+\mathsf RQ_s\mathsf R^\dagger$ and $W_s=\diag(\sqrt{r_0},\ldots,\sqrt{r_7})$, using \eqref{eq:raw-errors}. The proposed simulator Choi matrix is
\begin{equation}
 J_{\cD_s}=(W_s^{-1}\otimes I)H_s(W_s^{-1}\otimes I)\succeq0.
 \label{eq:simulator-choi}
\end{equation}
Trace preservation and the transpose-simulation identity reduce to
\begin{align}
 \Tr_BH_s&=\diag(2s,1,s,s,3s,3s,2s^2-1,2s^2-1),\notag\\
 \sum_{e,f}(L_f^\dagger L_e)_{ji}(H_s)_{ef}
 &=\left[\sum_e r_eL_e\ket i\bra jL_e^\dagger\right]^{\mathsf T}
 \quad(0\le i,j\le3).
 \label{eq:simulator-identities}
\end{align}
Here $(H_s)_{ef}$ is the $4\times4$ block with environmental indices $e,f$. Substitution in \eqref{eq:simulator-gram} verifies both identities, the second on all sixteen input matrix units. The first gives $\Tr_BH_s=W_s^2$, hence $\Tr_BJ_{\cD_s}=I_E$.

\Needspace{6\baselineskip}
For composition, the Kraus weights cancel because
\begin{equation}
 [\cN_s^c(X)]_{ef}=\frac{\sqrt{r_er_f}}{\nu_s}\Tr(L_eXL_f^\dagger),\qquad
 \cD_s(\ket e\bra f)=\frac{(H_s)_{ef}}{\sqrt{r_er_f}}.
\end{equation}
Thus the second identity in \eqref{eq:simulator-identities}, divided by $\nu_s$, is exactly $\cD_s\cN_s^c(\ket i\bra j)=(\cN_s(\ket i\bra j))^{\mathsf T}$. Linearity extends this identity to every input operator. Together with Choi positivity and trace preservation, it proves that $\cD_s$ is a CPTP transpose-antidegrading map for $s\ge3$.

\Needspace{8\baselineskip}
\subsubsection{Complementary-channel support bounds}\label[appendix]{app:environment}
We prove \eqref{eq:env-bounds} by expressing every signal amplitude column as a linear combination of background columns. A bound on the coefficient matrix then gives the required operator inequality, including coherences between environmental labels. For $z=\sum_{a=0}^1\sum_{j=0}^3z_{aj}\ket{a,j}$, define the unweighted environmental amplitude matrix
\begin{equation}
 [\mathsf A(z)]_{(a,e),b}=\sum_{j=0}^3z_{aj}(L_e)_{bj},
 \qquad 0\le e\le7,\quad0\le b\le3.
\end{equation}
The column index $b$ is the receiver level traced out in the complement. Collect the eight background columns in $B_s=(b_0,\ldots,b_7)$ and the four signal columns in $S_s=(d_0,\ldots,d_3)$:
\begin{equation}
 B_s=\bigl(\mathsf A(\sqrt{s+1}\,u_s),\mathsf A(\sqrt{s+1}\,v_s)\bigr),\qquad
 S_s=\mathsf A(\sqrt2\,w).
\end{equation}
Substitution of the ladder errors gives the four relations
\begin{align}
 d_0&=-b_5/\sqrt s,&d_3&=-b_2/\sqrt s,\\
 d_1&=\frac{2\sqrt s}{s-1}b_0-\frac{s+1}{s-1}b_6,&
 d_2&=-\frac{s+1}{s-1}b_1+\frac{2\sqrt s}{s-1}b_7.
 \label{eq:environment-factor}
\end{align}
Thus $S_s=B_sT_s$ with orthogonal coefficient columns and
\begin{equation}
 T_s^\dagger T_s=\diag\!\left(\frac1s,\frac{s^2+6s+1}{(s-1)^2},
 \frac{s^2+6s+1}{(s-1)^2},\frac1s\right).
\end{equation}
Before inserting Kraus weights, the background and signal operators are $\widetilde\eta_0=B_sB_s^\dagger/[2(s+1)]$ and $\widetilde\eta_1=S_sS_s^\dagger/2$. Hence
\begin{equation}
 \widetilde\eta_1\preceq(s+1)\norm{T_s}_\infty^2\widetilde\eta_0
 =c_s\widetilde\eta_0\quad(s\ge3).
\end{equation}
Congruence by $(I_R\otimes W_s)/\sqrt{\nu_s}$ gives $\eta_1\preceq c_s\eta_0$, including environmental-label coherences.

When Eve does not receive $R$, tracing out that register gives a stronger bound with constant $s$. The reduced unweighted matrices are diagonal, with entries
\begin{align}
 \widetilde\eta_0^E:&\quad
 \left(1,s,\frac{2s+1}{2(s+1)},\frac{2s+1}{2(s+1)},\frac12,\frac12,
 \frac{s}{2(s+1)},\frac{s}{2(s+1)}\right),\notag\\
 \widetilde\eta_1^E:&\quad
 \left(1,s^2,\frac{s+1}{2},\frac{s+1}{2},\frac12,\frac12,0,0\right).
\end{align}
Their coordinate ratios are at most $s$ for $s\ge3$, using $(s+1)^2/(2s+1)\le s$. Congruence by $W_s/\sqrt{\nu_s}$ gives $\eta_1^E\preceq s\eta_0^E$. Both bounds also give the support containment needed in \cref{prop:support-activation}.

\clearpage
\phantomsection
\addtocontents{toc}{\protect\setcounter{tocdepth}{2}}
\addcontentsline{toc}{section}{References}
\bibliographystyle{unsrtnat}
\bibliography{references}

@article{Personick1971,
  author={Personick, Stewart D.},
  title={Application of Quantum Estimation Theory to Analog Communication over Quantum Channels},
  journal={IEEE Transactions on Information Theory}, volume={17}, number={3}, pages={240--246}, year={1971},
  doi={10.1109/TIT.1971.1054643}
}

@Article{watanabe2012private,
  author    = {Watanabe, Shun},
  journal   = {Physical Review A},
  title     = {Private and quantum capacities of more capable and less noisy quantum channels},
  year      = {2012},
  number    = {1},
  pages     = {012326},
  volume    = {85},
  doi       = {10.1103/PhysRevA.85.012326},
  publisher = {APS},
}

@article{TemmeVerstraete2014,
  author={Temme, Kristan and Verstraete, Frank},
  title={Quantum Chi-Squared and Goodness of Fit Testing},
  year={2015}, journal={Journal of Mathematical Physics}, volume={56}, pages={012202},
  url={https://arxiv.org/abs/1112.6343v2}
}

@article{BraunsteinCaves1994,
  author={Braunstein, Samuel L. and Caves, Carlton M.},
  title={Statistical Distance and the Geometry of Quantum States},
  journal={Physical Review Letters}, volume={72}, pages={3439--3443}, year={1994},
  doi={10.1103/PhysRevLett.72.3439}
}

@misc{MandalPfister2026,
  author={Mandal, Avijit and Pfister, Henry D.},
  title={{Reed--Muller} Codes on {CQ} Channels via a New Correlation Bound for Quantum Observables},
  year={2026}, howpublished={arXiv:2502.03785v3},
  url={https://arxiv.org/abs/2502.03785v3}
}

@misc{ZhuWang2026,
  author={Zhu, Chengkai and Wang, Xin},
  title={Quantum Incapacity beyond No-Cloning and {PPT} Mechanisms},
  year={2026}, howpublished={arXiv:2607.24693v2},
  url={https://arxiv.org/abs/2607.24693v2}
}

@misc{privateCapZero2026,
  author={{Uthirakalyani G} and Halder, Pritam and Elkouss, David},
  title={Private communication from {Pauli} channels with no privacy},
  year={2026}, howpublished={arXiv:2609.20133v1},
  url={https://arxiv.org/abs/2609.20133v1}
}

@article{HircheRouzeFranca2022,
  author={Hirche, Christoph and Rouz{\'e}, Cambyse and {Stilck Fran\c{c}a}, Daniel},
  title={On Contraction Coefficients, Partial Orders and Approximation of Capacities for Quantum Channels},
  journal={Quantum}, volume={6}, pages={862}, year={2022},
  doi={10.22331/q-2022-11-28-862}, url={https://arxiv.org/abs/2011.05949v4}
}

@misc{HircheGeorgeNuradhaWilde2026,
  author={Hirche, Christoph and George, Ian and Nuradha, Theshani and Wilde, Mark M.},
  title={Conditional contraction coefficients and their applications to quantum networks},
  year={2026}, howpublished={arXiv:2608.27171v1},
  url={https://arxiv.org/abs/2608.27171v1}
}

@misc{ChengHircheTomamichel2026,
  author={Cheng, Hao-Chung and Hirche, Christoph and Tomamichel, Marco},
  title={The exact strong converse exponent for private communication over quantum channels},
  year={2026}, howpublished={arXiv:2609.38144v1},
  url={https://arxiv.org/abs/2609.38144v1}
}

@misc{WildePrivateStrongConverse2026,
  author={Wilde, Mark M.},
  title={An exponential strong converse for private communication over degradable quantum channels},
  year={2026}, howpublished={arXiv:2609.32515v1},
  url={https://arxiv.org/abs/2609.32515v1}
}

@article{BelzigGaoSmithWu2025,
  author={Belzig, Paula and Gao, Li and Smith, Graeme and Wu, Peixue},
  title={Reverse-Type Data Processing Inequality},
  journal={Communications in Mathematical Physics}, volume={406}, pages={295}, year={2025},
  doi={10.1007/s00220-025-05474-4}, url={https://arxiv.org/abs/2411.19890v3}
}

@article{PetzToth1993,
  author={Petz, D{\'e}nes and T{\'o}th, G{\'e}za},
  title={The {Bogoliubov} Inner Product in Quantum Statistics},
  journal={Letters in Mathematical Physics}, volume={27}, pages={205--216}, year={1993},
  doi={10.1007/BF00739578}
}

@article{LesniewskiRuskai1999,
  author={Lesniewski, Andrew and Ruskai, Mary Beth},
  title={Monotone {Riemannian} Metrics and Relative Entropy on Noncommutative Probability Spaces},
  journal={Journal of Mathematical Physics}, volume={40}, pages={5702--5724}, year={1999},
  url={https://arxiv.org/abs/math-ph/9808016}
}

@misc{BlumeKohoutHayden2006,
  author={Blume-Kohout, Robin and Hayden, Patrick},
  title={Accurate Quantum State Estimation via ``Keeping the Experimentalist Honest''},
  year={2006}, howpublished={arXiv:quant-ph/0603116},
  url={https://arxiv.org/abs/quant-ph/0603116}
}

@misc{Frongillo2022,
  author={Frongillo, Rafael},
  title={Quantum Information Elicitation},
  year={2022}, howpublished={arXiv:2203.07469},
  url={https://arxiv.org/abs/2203.07469}
}

@misc{RegulaTomamichel2026,
  author={Regula, Bartosz and Tomamichel, Marco},
  title={Rethinking Quantum Smooth Entropies: Tight One-Shot Analysis of Quantum Privacy Amplification},
  year={2026}, howpublished={arXiv:2603.04493v4},
  url={https://arxiv.org/abs/2603.04493v4}
}

@misc{AnshuBertaJainTomamichel2018,
  author={Anshu, Anurag and Berta, Mario and Jain, Rahul and Tomamichel, Marco},
  title={Partially Smoothed Information Measures},
  year={2018}, howpublished={arXiv:1807.05630},
  url={https://arxiv.org/abs/1807.05630}
}

@article{TsangAlbarelliDatta2020,
  author={Tsang, Mankei and Albarelli, Francesco and Datta, Animesh},
  title={Quantum Semiparametric Estimation},
  journal={Physical Review X}, volume={10}, pages={031023}, year={2020},
  doi={10.1103/PhysRevX.10.031023}, url={https://arxiv.org/abs/1906.09871}
}

@misc{Buscemi2016,
  author={Buscemi, Francesco},
  title={Degradable Channels, Less Noisy Channels, and Quantum Statistical Morphisms: An Equivalence Relation},
  year={2016}, howpublished={arXiv:1511.08893},
  url={https://arxiv.org/abs/1511.08893}
}

@misc{BuscemiDattaStrelchuk2014,
  author={Buscemi, Francesco and Datta, Nilanjana and Strelchuk, Sergii},
  title={Game-Theoretic Characterization of Antidegradable Channels},
  year={2014}, howpublished={arXiv:1404.0277},
  url={https://arxiv.org/abs/1404.0277}
}

@article{Devetak2005,
  author={Devetak, Igor},
  title={The Private Classical Capacity and Quantum Capacity of a Quantum Channel},
  journal={IEEE Transactions on Information Theory}, volume={51}, number={1}, pages={44--55}, year={2005},
  url={https://arxiv.org/abs/quant-ph/0304127}
}

@book{Wilde2017,
  author={Wilde, Mark M.},
  title={Quantum Information Theory},
  edition={Second},
  publisher={Cambridge University Press},
  year={2017},
  doi={10.1017/9781316809976}
}

@article{Penrose1955,
  author={Penrose, Roger},
  title={A Generalized Inverse for Matrices},
  journal={Mathematical Proceedings of the Cambridge Philosophical Society},
  volume={51}, number={3}, pages={406--413}, year={1955},
  doi={10.1017/S0305004100030401}
}

@misc{ZhuWangActivation2026,
 author={Zhu, Chengkai and Wang, Xin},
 title={Private Communication via Zero-Private-Capacity Quantum Channels},
 year={2026}, howpublished={arXiv:2609.10520v1},
 url={https://arxiv.org/abs/2609.10520v1}
}

@article{BradlerDutilHaydenMuhammad2010,
 author={Br{\'a}dler, Kamil and Dutil, Nicolas and Hayden, Patrick and Muhammad, Abubakr},
 title={Conjugate Degradability and the Quantum Capacity of Cloning Channels},
 journal={Journal of Mathematical Physics},volume={51},pages={072201},year={2010},
 url={https://arxiv.org/abs/0909.3297}
}

@article{HircheLeditzky2023,
 author={Hirche, Christoph and Leditzky, Felix},
 title={Bounding Quantum Capacities via Partial Orders and Complementarity},
 journal={IEEE Transactions on Information Theory},volume={69},pages={283--297},year={2023},
 url={https://arxiv.org/abs/2202.11688}
}

@article{RoofehKarimipour2024,
 author={Roofeh, S. and Karimipour, V.},
 title={The Noisy {Werner--Holevo} Channel and Its Properties},
 journal={Physical Review A},volume={109},pages={052620},year={2024},
 url={https://arxiv.org/abs/2310.15353}
}

@misc{SolymosJakabZimboras2024,
 author={Solymos, A. and Jakab, D. and Zimbor{\'a}s, Z.},
 title={Extendibility of {Brauer} States},year={2024},
 howpublished={arXiv:2411.04597},url={https://arxiv.org/abs/2411.04597}
}

@article{MyhrLutkenhaus2009,
 author={Myhr, Geir Ove and L{\"u}tkenhaus, Norbert},
 title={Spectrum Conditions for Symmetric Extendible States},
 journal={Physical Review A},volume={79},pages={062307},year={2009},
 url={https://arxiv.org/abs/0812.3667}
}

@misc{MullerHermesSingh2022,
 author={M{\"u}ller-Hermes, Alexander and Singh, Satvik},
 title={Bi-{PPT} Channels Are Entanglement Breaking},year={2022},
 howpublished={arXiv:2204.01685},url={https://arxiv.org/abs/2204.01685}
}

@article{SmithYard2008,
 author={Smith, Graeme and Yard, Jon},
 title={Quantum Communication With Zero-Capacity Channels},
 journal={Science}, volume={321}, pages={1812--1815}, year={2008},
 doi={10.1126/science.1162242}, eprint={0807.4935}, archivePrefix={arXiv}
}

@article{KnillLaflamme1997,
 author={Knill, Emanuel and Laflamme, Raymond},
 title={Theory of Quantum Error-Correcting Codes},
 journal={Physical Review A},volume={55},pages={900--911},year={1997},
 doi={10.1103/PhysRevA.55.900},url={https://arxiv.org/abs/quant-ph/9604034}
}

@article{LiWinterZouGuo2009,
  author={Li, Ke and Winter, Andreas and Zou, XuBo and Guo, GuangCan},
  title={The private capacity of quantum channels is not additive},
  journal={Physical Review Letters}, volume={103}, pages={120501}, year={2009},
  doi={10.1103/PhysRevLett.103.120501},
  eprint={0903.4308}, archivePrefix={arXiv}, url={https://arxiv.org/abs/0903.4308}
}

@article{SmithSmolin2009,
  author={Smith, Graeme and Smolin, John A.},
  title={Extensive nonadditivity of privacy},
  journal={Physical Review Letters}, volume={103}, pages={120503}, year={2009},
  doi={10.1103/PhysRevLett.103.120503},
  eprint={0904.4050}, archivePrefix={arXiv}, url={https://arxiv.org/abs/0904.4050}
}

@article{CaiWinterYeung2004,
  author={Cai, N. and Winter, A. and Yeung, R. W.},
  title={Quantum privacy and quantum wiretap channels},
  journal={Problems of Information Transmission}, volume={40}, number={4},
  pages={318--336}, year={2004},
  doi={10.1007/s11122-005-0002-x},
  url={https://link.springer.com/article/10.1007/s11122-005-0002-x}
}

@inproceedings{ElkoussStrelchuk2015,
  author={Elkouss, David and Strelchuk, Sergii},
  title={Quantum Capacity Can Be Greater Than Private Information for Arbitrarily Many Uses},
  booktitle={10th Conference on the Theory of Quantum Computation, Communication and Cryptography},
  series={Leibniz International Proceedings in Informatics}, volume={44},
  pages={64--72}, year={2015},
  publisher={Schloss Dagstuhl--Leibniz-Zentrum f{\"u}r Informatik},
  doi={10.4230/LIPIcs.TQC.2015.64},
  url={https://drops.dagstuhl.de/entities/document/10.4230/LIPIcs.TQC.2015.64}
}

@article{HircheTomamichel2024,
  author={Hirche, Christoph and Tomamichel, Marco},
  title={Quantum {R{\'e}nyi} and $f$-Divergences from Integral Representations},
  journal={Communications in Mathematical Physics}, volume={405}, pages={208},
  year={2024}, doi={10.1007/s00220-024-05087-3},
  url={https://arxiv.org/abs/2306.12343v3}
}

@article{BertaFawziTomamichel2017,
  author={Berta, Mario and Fawzi, Omar and Tomamichel, Marco},
  title={On variational expressions for quantum relative entropies},
  journal={Letters in Mathematical Physics}, volume={107}, number={12},
  pages={2239--2265}, year={2017}, doi={10.1007/s11005-017-0990-7},
  url={https://arxiv.org/abs/1512.02615v2}
}

@article{Smith2008Symmetric,
 author={Smith, Graeme},
 title={Private Classical Capacity with a Symmetric Side Channel and Its Application to Quantum Cryptography},
 journal={Physical Review A}, volume={78}, pages={022306}, year={2008},
 doi={10.1103/PhysRevA.78.022306}, url={https://arxiv.org/abs/0705.3838}
}

@article{Sion1958,
  author = {Maurice Sion},
  title = {On general minimax theorems},
  journal = {Pacific Journal of Mathematics},
  volume = {8},
  number = {1},
  pages = {171--176},
  year = {1958},
  doi = {10.2140/pjm.1958.8.171}
}
\end{document}